\documentclass[superscriptaddress,showpacs,floatfix,nofootinbib,nobibnotes,aps,prl,twocolumn,amsmath,amssymb,10pt]{revtex4-1}
\usepackage{amsthm,dsfont}
\usepackage{bm}
\usepackage{lmodern}
\usepackage{tcolorbox}
\usepackage[tracking=true]{microtype}
\usepackage[english]{babel}
\usepackage{placeins}
\usepackage{hyperref}
\usepackage{subcaption}
\usepackage[normalem]{ulem}

\usepackage{tikz}
\usetikzlibrary{arrows.meta,calc,arrows}
\usepackage{pgfplots}
\pgfplotsset{width=7cm,compat=1.3}

\newcommand{\tf}{\hat{f}}

\newcommand{\R}{\mathbb{R}}
\newcommand{\N}{\mathcal{N}}
\newcommand{\eps}{\varepsilon}
\renewcommand{\P}{\mathbb{P}}

\newcommand{\E}{\mathbb{E}}
\newcommand{\Var}{\mathrm{Var}}
\newcommand{\Cov}{\mathrm{Cov}}

\newcommand{\ed}{\stackrel{d}{=}}
\newcommand{\sign}{\mathrm{sign}}

\newcommand{\Pl}{\text{pr}^{(\lambda)}}

\newcommand{\bW}{\mathbf{W}}
\newcommand{\bB}{\mathbf{B}}
\newcommand{\bL}{\mathbf{L}}
\newcommand{\bC}{\mathbf{C}}
\newcommand{\bI}{\mathbf{I}}
\newcommand{\bS}{\mathbf{S}}
\newcommand{\bJ}{\mathbf{J}}
\newcommand{\bV}{\mathbf{V}}
\newcommand{\bR}{\mathbf{R}}
\newcommand{\bphi}{\bm{\phi}^{(\ell)}}
\newcommand{\Sp}{\bm{\Sigma}_p}
\newcommand{\bSigma}{\bm{\Sigma}}
\newcommand{\Sf}{\bm{\Sigma}_f}

\newcommand{\bhe}{\hat{\mathbf{e}}}
\newcommand{\bv}{\mathbf{v}}
\newcommand{\btheta}{\bm{\theta}}
\newcommand{\bx}{\mathbf{x}}
\newcommand{\bX}{\mathbf{X}}
\newcommand{\bZ}{\mathbf{Z}}
\newcommand{\bs}{\mathbf{s}}
\newcommand{\bp}{\mathbf{p}}
\newcommand{\bpell}{\mathbf{p}^{(\ell)}}
\newcommand{\bmu}{\bm{\mu}}
\newcommand{\by}{\mathbf{y}}

\newcommand{\fell}{f^{(\ell)}}
\newcommand{\bfl}{\mathbf{f}}
\newcommand{\bfell}{\mathbf{f}^{(\ell)}}
\newcommand{\tbf}{\hat{\mathbf{f}}}
\newcommand{\bnu}{\bm{\nu}}

\newcommand{\bfmostlkl}{\widetilde{\mathbf{f}}^{(\ell)}}
\newcommand{\fmostlkl}{\widetilde{f}^{(\ell)}}

\newcommand{\fmostlikely}{\widetilde{f}^{(\ell)}}
\newcommand{\bfmostlikely}{\widetilde{\mathbf{f}}^{(\ell)}}
\newcommand{\frv}{\bar{\mathbf{f}}^{(\ell)}}
\newcommand{\prv}{\mathbf{\bar{p}}^{(\ell)}}

\xdefinecolor{dred}{RGB}{205,0,57}
\xdefinecolor{dgreen}{RGB}{25,205,25}
\xdefinecolor{dorange}{RGB}{255,100,0}

\DeclareMathOperator*{\argmin}{arg\,min}
\DeclareMathOperator*{\arginf}{arg\,inf}

\addto\captionsenglish{}

\newtheorem{innercustomthm}{Proposition}
\newenvironment{customthm}[1]
  {\renewcommand\theinnercustomthm{#1}\innercustomthm}
  {\endinnercustomthm}

\newtheorem{lem}{Lemma}

\usepackage[utf8x]{inputenc}
\usepackage{textcomp}

\begin{document}

\title{Emergent failures and cascades in power grids: a statistical physics perspective}
\author{Tommaso Nesti}
\affiliation{CWI, Amsterdam 1098 XG, Netherlands}
\author{Alessandro Zocca}
\affiliation{California Institute of Technology, Pasadena, California 91125, USA}
\author{Bert Zwart}
\affiliation{CWI, Amsterdam 1098 XG, Netherlands}
\date{\today}

\begin{abstract}
%
We model power grids transporting electricity generated by intermittent renewable sources as complex networks, where line failures can emerge indirectly by noisy power input at the nodes.
By combining concepts from statistical physics and the physics of power flows, and taking weather correlations into account, we rank line failures according to their likelihood and establish the most likely way such failures occur and propagate. Our insights are mathematically rigorous in a small-noise limit and are validated with data from the German transmission grid.
\end{abstract}

\maketitle
Understanding cascading failures in complex networks is of great importance and has received a lot of attention in recent years~\cite{Albert2004,AB02,AlbertJeongBarabasi2000,Crucitti2003,Crucitti2003a,Mirzasoleiman2011,
Motter2002,Motter2004,Heide2008,Kinney2005,Schafer2018,Sun2007,Yang2017,Watts2002,Witthaut2015,Witthaut2016,Witthaut2013}. 
Despite proposing different mechanisms for their evolution, a common feature is that cascades are triggered by some \textit{external} event. This initial attack is chosen either (i) deliberately, to target the most vulnerable or crucial network component or (ii) uniformly at random,  to understand the average network reliability. This distinction led to the insight that complex networks are resilient to random attacks, but vulnerable to targeted attacks~\cite{Cohen2000,Cohen2001,Motter2002}. However, both lead to the \textit{direct} failure of the attacked network component.

In this Letter, we focus on networks in which edge failures occur in a fundamentally different manner. Specifically, we consider networks where fluctuations of the node inputs can trigger edge failures. 
The realization (which we call configuration) of the noise at the nodes is not only the cause of edge failures, but can also impact the way they propagate in the network.

We present our results in the context of power grids that transport electricity generated by solar and wind parks. In power grids, line failures can arise when the network is driven from a stable state to a critically loaded state by external factors;
intermittent power generation at the nodes causes random fluctuations in the line power flows, possibly triggering outages and cascading failures.
Thus, line failures can emerge {\em indirectly} due to the interplay between noisy correlated (due to weather) power input at the nodes, the network structure, and power flow physics.
This interplay is challenging to analyze, yet this problem is urgent as the penetration of renewable energy sources is increasing~\cite{Bienstock2015, Dobson2007}.

We analyze this interplay using statistical physics and large deviations theory.
We consider a parsimonious static stochastic model similar to~\cite{Wang2012}, introduce a scaling parameter $\eps$ describing the magnitude of the noise and consider the regime $\eps \to 0$.
In the limit, we can identify the most vulnerable lines and explicitly determine the most likely configuration of power inputs leading to failures and subsequent propagating failures. These results 
are validated 
using real data for the German transmission network. 
Previous works applying large-deviations techniques to problems in complex networks dynamics, such as epidemic extinction and biophysical networks, include~\cite{Wells2015,Hindes2016}.


We model a transmission network by a connected graph $G$ with $n$ nodes representing the \textit{buses} and $m$ directed edges modeling \textit{transmission lines}. 
The nominal values of net power injections at the nodes are given by $\bmu=\{\mu_i\}_{i=1,\ldots,n}$.
We model the stochastic fluctuation of the power injections around $\bmu$, due to variability in renewable generation, by means of the random vector $\bp=\{p_i\}_{i=1,\ldots,n}$, which is assumed to follow a multivariate Gaussian distribution with density
\begin{equation}
\label{eq:ass}
        \varphi(\bx)=\frac{\exp (-{\scriptstyle \frac{1}{2}}(\bx-\bmu)^T (\eps\Sp )^{-1}(\bx-\bmu))}{(2\pi)^{\frac{n}{2}} \det (\eps \Sp )^{ \frac{1}{2}}},
\end{equation}
with $\eps\Sp \in\R^{n\times n}$ being the covariance matrix of $\bp$.
In our theoretical analysis, we assume that $\Sp$ is known and let $\eps \to 0$.


%
The Gaussian assumption is debatable, both for solar and wind. While consistent with atmospheric physics~\cite{BCH14} and recent wind park statistics~\cite{Kolumban2017,Troldborg2014}, different models are preferred for different timescales~\cite{Brouwer2014,Schlachtberger2016,Peings2014,MilanWaechterPeinke2013}. An extension of our framework to the dynamic model in~\cite{MilanWaechterPeinke2013} looks promising (using Freidlin-Wentzell theory as in~\cite{NestiNairZwart}). For a static non-Gaussian extension, see~\cite{NZZ17sm}.

\begin{figure*}[ht!]
    \centering
    \begin{subfigure}[t]{0.3\textwidth}
        \centering
\includegraphics[width=\textwidth,height=4.8cm]{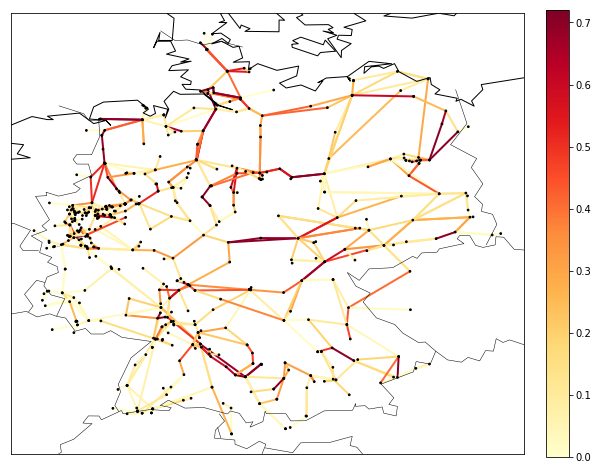}
\caption{\footnotesize   Nominal line flows  $|\nu_{\ell}|$ at $11$am.}\label{fig:nominalflows_11am}
    \end{subfigure}%
    ~
    \begin{subfigure}[t]{0.3\textwidth}
        \centering
\includegraphics[width=\textwidth,height=4.8cm]{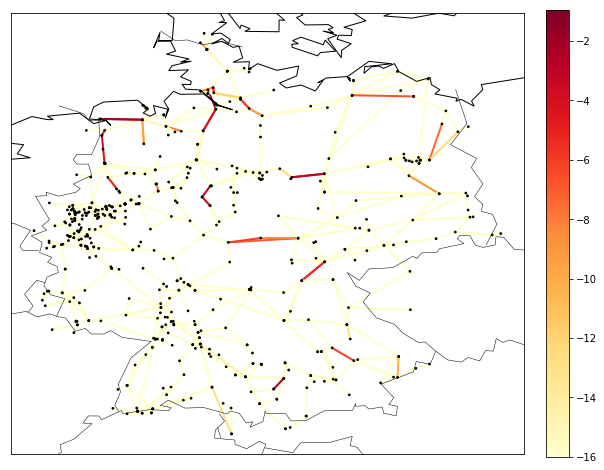}
\caption{\footnotesize True overload probabilities $\log_{10}\mathbb{P}(|f_{\ell}|\ge 1)$ at $11$am.
}
\label{fig:prob_11am}
    \end{subfigure}
    ~
     \begin{subfigure}[t]{0.3\textwidth}
     \centering
      \includegraphics[width=\textwidth,height=4.8cm]{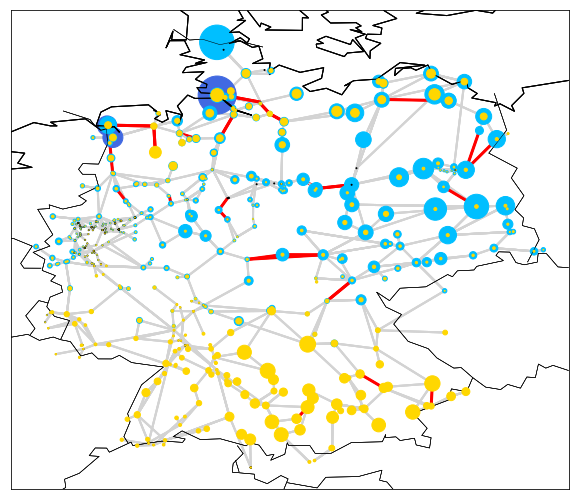}
\caption{\footnotesize
Top $5\%$ of most likely lines to fail (red) at $11$am, according to~\eqref{eq:decayrates}, and nominal injections from renewable sources. 
}
\label{fig:5percent_mostlikely_mix_11am}
  \end{subfigure}
  \label{fig:panel}
  \caption{}\vspace{-0.5cm}
  \end{figure*}



Assuming the vector $\bmu$ has zero sum and using the DC approximation~\cite{Bienstock2015}, the line power flows $\bfl=\{f_i\}_{i=1,\ldots,m}$ are given by
\begin{equation}
\label{eq:linearmap}
	\bfl=\bV \bp,
\end{equation}
where $\bV$ is an $m\times n$ matrix encoding the grid topology and parameters (i.e.,~line susceptances).
The DC approximation is commonly used in transmission system analysis~\cite{Purchala2005,Stott2009,Powell2004,Wood2014}. More realistic nonlinear models based on AC power flows~\cite{Mehta2016} may be analyzed leveraging the contraction principle \cite{DemboZeitouni2009}. 

The total net power injected in the network $\sum_{i=1}^n p_i$ is non-zero as $\bp$ is random.
Automated affine response and redispatch mechanisms take care of this issue in power grids. Mathematically, this corresponds to a ``distributed slack'' in our model: the total power injection mismatch is distributed uniformly among all nodes (the matrix $\bV$ accounts for this; see~\cite{NZZ17sm}).

In view of Eqs.~\eqref{eq:ass}-\eqref{eq:linearmap}, the line power flows $\bfl$ also follow a multivariate Gaussian distribution with mean $\bnu$ and covariance matrix $\eps \Sf$. The vector $\bnu =\bV \bmu\in\R^m$ describes the nominal line flows, while the covariance matrix $\eps \Sf =\eps \bV \Sp  \bV^T$ describes the correlations between line flows fluctuations, taking into account both the correlations of the power injections (encoded by $\Sp $) and correlations created by the network topology due to power flow physics (Kirchhoff's laws) via $\bV$.

A line \textit{overloads} if the absolute amount of power flowing in it exceeds a given \textit{line threshold}. We assume that such overloads immediately lead to the outage of the corresponding line, to which we will henceforth refer simply as \textit{line failure}.
The rationale behind this assumption is that there are security relays on high voltage transmission lines performing an emergency shutdown as soon as the current exceeds a dangerous level. Without such mechanisms, lines may overheat, sag and eventually trip.

We can express the line flows in units of the line threshold by incorporating the latter in the definition of $\bV$~\cite{NZZ17sm}, so that $\bfl$ is the vector of \textit{normalized} line power flows and the failure of line $\ell$ corresponds to $|f_\ell| \ge 1$. We let the power grid operate on average safely by assuming that $\max_{\ell=1,\dots,m}|\nu_{\ell}|<1$, so that only large fluctuations of line flows lead to failures.

We are most interested in scenarios where power grids are highly stressed, meaning that the nominal power injections $\{\mu_i\}_{i=1,\dots,n}$ are such that the corresponding nominal line power flows $\{\nu_\ell\}_{\ell=1,\dots,m}$ are close to their thresholds. Such a stress could be caused by very high wind generation~\cite{Pesch2014}.

An illustrative scenario is reported in Fig.~\ref{fig:nominalflows_11am}, which depicts a snapshot of nominal line flows on the SciGRID German network~\cite{SCIGRID00}. SciGRID is a detailed model of the actual German transmission network with $n=585$ buses and $m=852$ lines that we use as main illustration. The dataset includes load/generation time series, line limits, grid topology and generation costs.
In our case study, we obtain $\bmu$ by solving an Optimal Power Flow problem (OPF~\cite{Huneault1991}) based on realistic data for wind and solar generation, and we estimate $\eps\Sp$ using ARMA models; for details see the supplement~\cite{NZZ17sm}, which also describes a setting covering conventional controllable power plants.


We now turn to the analysis of emergent failures and their propagation using large deviations theory~\cite{Touchette2009}. We begin by deriving the exponential decay of probabilities of single line failure events $|f_\ell| \ge 1$ for $\ell=1,\dots,m$.
As line power flows are Gaussian, we obtain, see Example 3.1 in~\cite{Touchette2009}, that
\begin{equation}\label{eq:decayrates}
	I_\ell  =-\lim_{\eps\to 0} \eps \log \P_\eps (|f_\ell|\ge 1) = \frac{(1-|\nu_\ell|)^2}{2 \sigma^2_\ell},
\end{equation}
where $\sigma^2_\ell = (\Sf)_{\ell\ell}$. We call $I_\ell$ the {\em decay rate} of the failure probability of line $\ell$.
Thus, for small $\eps$, we approximate the probability of the emergent failure of line $\ell$ as
\begin{equation}\label{eq:individual_LD_approx}
	\P(|f_\ell| \ge 1) \approx \exp (-I_\ell / \eps)=\exp \Bigl(- \frac{(1-|\nu_\ell|)^2}{2 \eps\sigma^2_\ell} \Bigr),
\end{equation}

and that of the first emergent failure as
\begin{equation}\label{eq:LD_approx}
	\P(\max_\ell |f_\ell| \ge 1) \approx \exp (- \min_\ell I_\ell / \eps).
\end{equation}
These approximations for failure probabilities may not be sharp in general, even when $\eps$ is small, since all terms that are decaying subexponentially in $1/\eps$ are ignored.
Nevertheless, Eq.~\eqref{eq:individual_LD_approx} is quite useful for ranking purposes,
%
allowing to explicitly identify the lines that are most likely to fail.
To verify this empirically, we note that the expression in Eq.~\eqref{eq:individual_LD_approx} only depends on the product $\eps\sigma^2_{\ell}=\eps(\bV\Sp \bV^T)_{\ell\ell}$, and thus, ultimately, only on the product $\eps \Sp $, which in our case study we estimate directly from the SciGRID data, see~\cite{NZZ17sm}.

Fig.~\ref {fig:prob_11am} shows the heatmap for the exact line failure probabilities $\mathbb{P}(|f_{\ell}|\ge 1)$, for the same day and hour as in Fig.~\ref{fig:nominalflows_11am}: it is clear that a larger $|\nu_{\ell}|$ does not necessarily imply a higher chance of failure.  Fig.~\ref{fig:5percent_mostlikely_mix_11am} depicts the $5\%$ most likely lines to fail, ranked according to $I_{\ell}$. 
The ranking based on the large deviations approximation successfully recovers the most likely lines to fail, and, in fact, yields the same ordering as the one based on exact probabilities~\cite{NZZ17sm}, thus providing an accurate indicator of system vulnerabilities.

Fig.~\ref{fig:5percent_mostlikely_mix_11am} also illustrates the nominal renewable generation mix: the buses housing stochastic power injections have different colors 
(blue/light blue for wind offshore/onshore, yellow for solar) and sizes proportional to the absolute values of the corresponding nominal injections. Many vulnerable lines are located where the most renewable energy production occurs. However, the interplay between network topology, power flows physics and correlation in power injections caused by weather fluctuations, results in a spread-out arrangement of vulnerable lines, which is hard to infer by looking at nominal values only.




We proceed with an analysis of how emergent failures occur, using again large deviations theory. In particular, we provide an  explicit estimate of the most
likely power injection that caused a specific emergent failure.
To this end, we fix a line $\ell$ and consider the conditional distribution of $\bp$, given $|f_{\ell}| \geq 1$.
The mean of this distribution greatly simplifies as $\eps \to 0$ to
\begin{equation}
		\bpell = \arginf_{\bp\in\mathbb{R}^n \,:\, |\bhe_\ell^T \bV \bp|\ge 1} \frac{1}{2} (\bp-\bmu)^T \Sp ^{-1} (\bp-\bmu).
\end{equation}
If $\nu_\ell \neq 0$, the solution is unique and reads
\begin{align}\label{eq:most_likely_power}
	\bpell	=\, &\bmu + \frac{(\sign(\nu_\ell)-\nu_{\ell})}{\sigma_{\ell}^2} \Sp  \bV^T\bhe_{\ell},
\end{align}
where $\sign(a)=1$ if $a\geq 0$ and $-1$ otherwise, and $\bhe_{\ell} \in \R^m$ is the $\ell$-th unit vector.
As $\eps\to 0$,
the conditional variance of $\bp$ given $|f_{\ell}| \geq 1$ decreases to $0$ exponentially fast in $1/\eps$, yielding that the conditional
distribution of $\bp$ given $|f_\ell|\geq$ 1 gets sharply concentrated around $\bp^{(\ell)}$~\cite{NZZ17sm}.

We interpret  $\bpell$ as the \textit{most likely} power injection profile, conditional on the failure of line $\ell$. 
The corresponding line power flow profile $\bfell =\bV\bpell$
is
\begin{align}\label{eq:most_likely_flows}
        \fell_k =\nu_{k}+\frac{(\sign(\nu_{\ell})-\nu_{\ell})}{\sigma_{\ell}^2}\Cov(f_{\ell},f_{k}), \quad \forall \, k\neq \ell.
\end{align}
As such, our framework provides more explicit information than the approach in~\cite{Chertkov2011}, which approximates the most likely way events happen using the mode, without leveraging large deviations.
In our validation experiments, we found that the error between $\bpell$ and $\E[\,\bp\, \vert \,|f_\ell| \ge 1]$ is typically less than $1\%$ of the nominal values~\cite{NZZ17sm}. A numerical illustration is given in Fig.~\ref{fig:common}.

A key finding is that an emergent line failure does not occur due to large fluctuations only in neighboring nodes, but as a cumulative effect of small unusual fluctuations in the entire network ``summed up'' by power flow physics, and correlations in renewable energy.
Such an emergent failure requires every line flow to be driven to an unusual state $\fell_k$, which deviates from the nominal value $\nu_k$ by an amount proportional to the covariance $\Cov(f_{\ell},f_{k})$, in view of Eq.~\eqref{eq:most_likely_power}.


We continue by investigating the propagation of failures, combining our results describing the most likely power injections configuration leading to the first failure, and the power flow redistribution in the network afterwards.
To this end, we first differentiate between different types of line failures,
by assessing whether the most likely way for failure of line $\ell$ to occur is as (i) an \textit{isolated failure}, if $\smash{|\fell_k| <1}$ for all line $k \neq \ell$, or (ii) a \textit{joint failure}, if there exists some other line $k \neq \ell$ such that $\smash{|\fell_k| \geq 1}$.

Any type of line failure(s) cause(s) a global redistribution of the line power flows according to Kirchhoff's laws, which could trigger further outages and cascades.
In our setting, the power injections configuration $\bpell$ redistributes across an altered network $\widetilde{G}^{(\ell)}$ (a subgraph of the original graph $G$) in which line $\ell$ (and possible other lines, in case of a joint failure) has been removed, increasing stress on the remaining lines. The way this redistribution happens on $\widetilde{G}^{(\ell)}$ is governed by power flow physics and we assume that it occurs instantaneously. Extending this to dynamic models~\cite{Simonson2008, Schafer2017} is a natural future topic, as transient oscillatory effects may severe the impact of line failures.

The power flow redistribution amounts to compute a new matrix $\widetilde{\bV}$ linking the power injections and the new power flows, which can be constructed analogously to $\bV$~\cite{NZZ17sm}. The most likely power flow configuration on $\widetilde{G}^{(\ell)}$ after redistribution is $\bfmostlkl = \widetilde{\bV} \bpell.$

\begin{figure*}[ht!]
    \centering
        \begin{subfigure}[t]{0.3\textwidth}
     \centering
      \includegraphics[width=\textwidth,height=4.8cm]{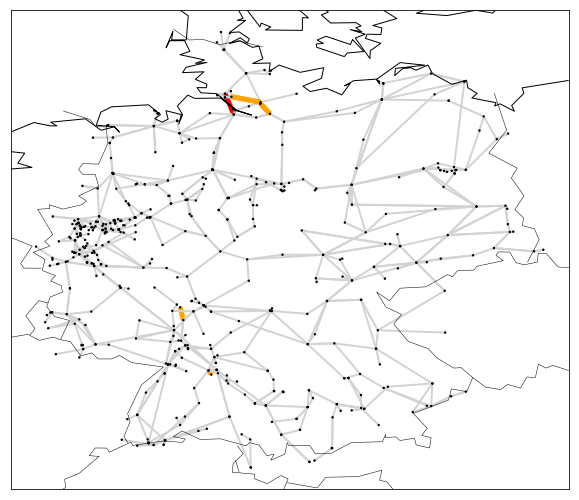}
\caption{\footnotesize  After the emergent failure of line $27$ (red) six additional lines (orange) fail, $4$pm. }
\label{fig:hamburg}
  \end{subfigure}
  ~
  \begin{subfigure}[t]{0.6\textwidth}
    \begin{subfigure}[t]{0.5\textwidth}
        \centering
\includegraphics[width=\textwidth,height=4.8cm]{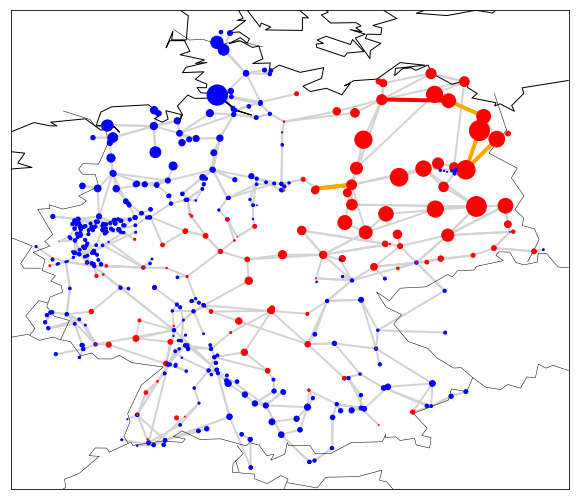}

    \end{subfigure}%
    ~
    \begin{subfigure}[t]{0.5\textwidth}
        \centering
\includegraphics[width=\textwidth,height=4.8cm]{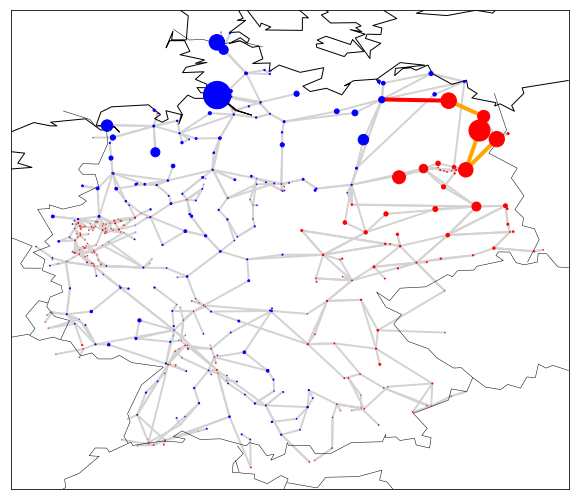}
    \end{subfigure}
\caption{\footnotesize  Most likely power injection $\bpell$ causing the isolated failure of line $720$ (red), and subsequent failures (orange). The bus sizes reflect how much $\bpell$ deviates from $\bmu$ at $11$am (red positive deviations, blue negative). Left, with correlation in noise); Right, without correlation in noise (setting to $0$ all the off-diagonals of $\bm{\Sigma}_p$).}
\label{fig:common}
 \end{subfigure}
  \label{fig:panel2}
  \caption{}\vspace{-0.5cm}
  \end{figure*}


In the special case of an isolated failure (say of line $\ell$) it is enough to calculate the vector $\bphi \in \R^{m-1}$ of (normalized) redistribution coefficients, known as \textit{line outage distribution factors} (LODF)~\cite{Guo2009}. The quantity $\phi^{(\ell)}_j$ takes values in $[-1,1]$, and $\smash{|\phi^{(\ell)}_j|}$ represents the percentage of power flowing in line $\ell$ that is redirected to line $j$ after the failure of the former. The most likely power flow configuration on $\widetilde{G}^{(\ell)}$ after redistribution then equals $\bfmostlkl = \{f^{(\ell)}_k\}_{k\neq \ell} + f^{(\ell)}_\ell\bphi ,$ where $\fell_\ell = \pm 1$ depending on the way the power flow is most likely to exceed the threshold $1$. The power flow configuration $\bfmostlkl$ can be efficiently used to determine which lines subsequently fail, by checking for which $k$ we have $\smash{|\fmostlkl_k| \geq 1}$, see~\cite{NZZ17sm}.
%


There is much evidence that failures propagate non-locally in power grids~\cite{Jung2016,Kettemann2016,Labavic2014,Ronellenfitsch2017,Manik2017}. To analyze this in our framework we first consider a ring network with $\bmu=0$ and $\Sp =I$. In this network there are two paths along which power can flow between any two nodes, using the convention that a positive flow corresponds to a counter-clockwise direction. If line $\ell$ fails, the power originally flowing on line $\ell$ must now flow on the remaining path in the opposite direction. To make this rigorous we show in~\cite{NZZ17sm} that $\smash{\phi^{(\ell)}_k=-1}$ for every $k\neq \ell$. As power flows must sum to zero by Kirchhoff's law, neighboring lines tend to have positively correlated power flows, while flows on distant lines exhibit negative correlations.
Hence, the power injections that make the power flows in line $\ell$ exceed the line threshold (say by becoming larger than $1$) also make the power flows in the antipodal half of the network negative. These will go beyond the line threshold $-1$ after the power flow redistributes, cf.\ Fig.~\ref{fig:cycle}.

\begin{figure}[!h]
\centering
\begin{tikzpicture}[scale=0.75]
	\begin{scope}[every node/.style={circle,thick}]
	    \node[fill=blue,minimum size=0.7cm] (A) at (0,0) {};
	    \node[fill=blue,minimum size=0.5421cm] (B) at (2,0) {};
	    \node[fill=blue,minimum size=0.3131cm] (C) at (3,1.732) {};
	    \node[fill=red,minimum size=0.3131cm] (D) at (2,3.464) {};
	    \node[fill=red,minimum size=0.5421cm] (E) at (0,3.464) {};
	    \node[fill=red,minimum size=0.7cm] (F) at (-1,1.732) {} ;
	\end{scope}
	
%
	
	\begin{scope}[>={Stealth[black]},every edge/.style={draw=gray,very thick}]
	    \path [-] (A) edge node {} (B);
	    \path [-] (B) edge node {} (C);
	    \path [-] (C) edge node {} (D);
	    \path [-] (D) edge node {} (E);
	    \path [-] (E) edge node {} (F);
	\end{scope}
	\begin{scope}[>={Stealth[black]},every edge/.style={draw=orange,very thick}]
	    \path [-] (F) edge node {} (A);
	\end{scope}

	\node[] at (-0.75,0.75) {\footnotesize ${\color{red}1}$};
	\node[] at (-0.775,2.75) {\footnotesize $1/7$};
	\node[] at (1,3.7) {\footnotesize -$13/35$};
	\node[] at (3.1,2.8) {\footnotesize -$19/35$};
	\node[] at (3.2,0.8) {\footnotesize -$13/35$};
	\node[] at (1,-0.3) {\footnotesize $1/7$};
	
	\begin{scope}[every node/.style={circle,thick}]
	    \node[fill=blue,minimum size=0.7cm] (G) at (6,0) {};
	    \node[fill=blue,minimum size=0.5421cm] (H) at (8,0) {};
	    \node[fill=blue,minimum size=0.3131cm] (I) at (9,1.732) {};
	    \node[fill=red,minimum size=0.3131cm] (L) at (8,3.464) {};
	    \node[fill=red,minimum size=0.5421cm] (M) at (6,3.464) {};
	    \node[fill=red,minimum size=0.7cm] (N) at (5,1.732) {} ;
	\end{scope}
	
%
	
	\node[] at (4.9,0.8) {\footnotesize failed};
	\node[] at (6.1,2.5) {\footnotesize $1/7${\color{red}-}${\color{red}1}$};
	\node[] at (7.1,3.7) {\footnotesize -$13/35${\color{red}-}${\color{red}1}$};
	\node[] at (9.3,2.8) {\footnotesize -$19/35${\color{red}-}${\color{red}1}$};
	\node[] at (9.4,0.8) {\footnotesize -$13/35${\color{red}-}${\color{red}1}$};
	\node[] at (7,-0.3) {\footnotesize $1/7${\color{red}-}${\color{red}1}$};
	
	\begin{scope}[>={Stealth[black]},every edge/.style={draw=gray,very thick}]
	    \path [-] (G) edge node {} (H);
	    \path [-] (M) edge node {} (N);
	\end{scope}
	\begin{scope}[>={Stealth[black]},every edge/.style={draw=orange,very thick}]
	    \path [-] (H) edge node {} (I);
	    \path [-] (I) edge node {} (L);
	    \path [-] (L) edge node {} (M);
	\end{scope}
	\begin{scope}[>={Stealth[black]},every edge/.style={dashed,draw=gray,very thick}]
	    \path [-] (N) edge node {} (G);
	\end{scope}
	
	\end{tikzpicture}
	\caption{\footnotesize Left: most likely power injections $\bpell$ leading to the failure of line $\ell$ (orange), visualized using the color and size of the nodes (red positive deviations, blue negative), together with power flows $\fell_k$. Right: situation after the power flow redistribution with three subsequent failures and the values $\fmostlkl_k=f^{(\ell)}_k - 1$, $k\neq \ell$. 
	}
	\label{fig:cycle}
\end{figure}
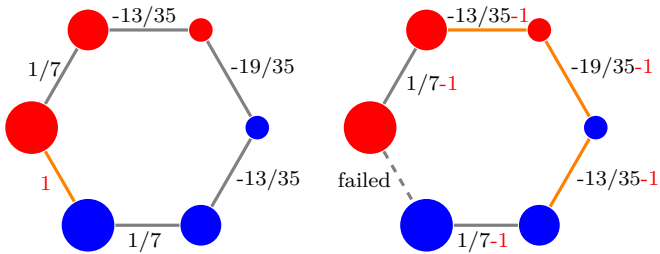


%



In the SciGRID example, Fig.~\ref{fig:hamburg} shows how the emergent isolated failure of line $\ell=27$ causes the failure of six more lines $k_1,\ldots,k_6$, two of which are far way from the original failure.
For validation purposes, we found numerically that $\mathbb{P}(\text{line $k_j$ fails }\,\forall j=1,\ldots,6\,\, \rvert\,\, |f_{27}|\ge 1)\ge  0.9987$.
Conversely, the failure of line $27$ under the nominal power injection profile leads to only two subsequent failures. The nontypical input caused other lines to be more loaded than expected, and these lines get more vulnerable as the cascades progresses, resulting in more subsequent failures.

To validate this insight, we have looked at the first two stages of emergent
cascading failures for several IEEE test networks, and compare them with those of
classical cascading failures, obtained using nominal power injection values rather than the most likely ones
and deterministic removal of the initial failing line; see \cite{NZZ17sm} for a precise description of the experiment.
As before, emergent cascades tend to lead to a higher number of subsequent failures in each stage.

A non-diagonal noise matrix $\Sp $ exacerbates these effects. Experiments (see Fig.~\ref{fig:common}) with our SciGRID case study suggest that, if there is a correlation in noise, for example due to fluctuations in weather patterns, the number of subsequent failures can become higher. Furthermore, it is easier for a failure to be triggered by many small disturbances across the network, compared to the case where these correlations are not taken into account. In the latter case, we see a more local effect with relatively larger disturbances.

In conclusion, we illustrated the potential of concepts from statistical physics and large deviations theory to analyze
emergent failures and their propagation in complex networks. 
Exogenous noise disturbances at the nodes, potentially amplified by correlations, push a complex network into a critical state in which edge failure may emerge. Large deviations theory provides a tool to rank such failures according to their likelihood and predicts how such failures most likely occur and propagate. 
When an emergent edge failure occurs, its impact on the network can be more significant than a purely exogenous failure, possibly resulting in cascades that propagate quicker than in classical vulnerability analysis.

The accuracy of the small noise limit has been validated in our case study, making the case for applying large deviations techniques to more realistic models.
In~ \cite{NZZ17sm} we propose a promising economic application of our approach, showing how our framework can shed light on the trade-off between network reliability and societal costs.
\\

\noindent
{\bf Acknowledgements.} We thank the referees for many useful comments, in particular for suggesting SciGRID. NWO Vici 639.033.413 and NWO Rubicon 680.50.1529 grants provided financial support. AZ acknowledges the support of Resnick Sustainability Institute at Caltech.

\nocite{CKvM16,Schaub2014,Soltan2015,Bapat2006,PyPSA2017,SCIGRID1,
SCIGRID2,OpenStreetMap2017,Milligan2003,Antonanzas2016,Huang2012,Zhang2013corr,
Hodge2011,Schweppe1988,Hines2013,Hines2015,Qi2015,YangNishikawaMotter2017}

%

\pagebreak
\clearpage

\onecolumngrid
\begin{center}
\textbf{\large Supplemental Material for:\\Emergent failures and cascades in power grids: a statistical physics perspective}
\end{center}
\vspace{1cm}

\twocolumngrid

\setcounter{equation}{0}
\setcounter{figure}{0}
\setcounter{table}{0}
\setcounter{page}{1}
\makeatletter
\renewcommand{\theequation}{S\arabic{equation}}
\renewcommand{\thefigure}{S\arabic{figure}}
\renewcommand{\bibnumfmt}[1]{[S#1]}

\section{Power grid model and DC approximation}\label{Power grid model and DC approximation}
We model the power grid network as a connected weighted graph $G$ with $n$ nodes, modeling buses, and $m$ edges, representing the transmission lines. We make use of the \textit{DC approximation}, which is commonly used in high-voltage transmission system analysis~\cite{Purchala2005,Stott2009,Powell2004,Wood2014}.

Choosing an arbitrary but fixed orientation of the transmission lines, the network structure is described by the \textit{edge-vertex incidence matrix} $\bC\in\R^{m\times n}$ defined as
\begin{equation*}
 	\bC_{\ell, i}=\begin{cases}
    	1		&\text{if } \ell=(i,j),\\
    	-1		&\text{if } \ell=(j,i),\\
    	0    	&\text{otherwise}.
    \end{cases}
\end{equation*}
Denote by $\beta_\ell=\beta_{i,j}=\beta_{j,i}>0$ the weight of edge $\ell=(i,j)$, corresponding to the \textit{susceptance} of that transmission line. By convention, we set $\beta_{i,j}=\beta_{j,i}=0$ if there is no transmission line between $i$ and $j$. Denote by $\bB$ the $m \times m$ diagonal matrix defined as $\bB=\mathrm{diag}(\beta_1,\dots, \beta_m)$.

The network topology and weights are simultaneously encoded in the \textit{weighted Laplacian matrix} of the graph $G$, defined as $\bL = \bC^\top \bB \bC$ or entry-wise as
\begin{equation*}
	L_{i,j} =\begin{cases}
		-\beta_{i,j}							& \text{if } i \neq j,\\
		\sum_{k\neq j} \beta_{i,k} 	& \text{if } i=j.
	\end{cases}
\end{equation*}
All the rows of $\bL$ sum up to zero and thus the matrix $\bL$ is singular. The eigenvalue zero has multiplicity one (thanks to the assumption that the graph $G$ is connected) and the corresponding eigenvector is $\bm{1}$. Denote by $\bv_2,\dots,\bv_n$ the remaining eigenvectors of $\bL$, which are orthogonal to $\bm{1}$ and thus have all zero sum.

According to the \textit{DC approximation}, the relation between any zero-sum vector of power injections $p \in \R^n$ and the phase angles $\btheta \in \R^n$ they induce in the network nodes can be written in matrix form as
\begin{equation*}
	\bp = \bL \btheta,
\end{equation*}
Defining $\bL^+ \in \R^{n \times n}$ as the \textit{Moore-Penrose} pseudo-inverse of $\bL$, we can rewrite this as
\begin{equation}
\label{eq:dcapprox}
	\btheta = \bL^+ \bp.
\end{equation}
This latter identity is particularly useful in our context, since it holds for any vector of power injections $\bp \in \R^n$, even if it has no zero sum. Indeed, decomposing the vector $p$ using the basis of eigenvectors $\bm{1}, \bv_2,\dots,\bv_n$ of $\bL^+$ one notices that the only component of $\bp$ with non-zero sum belongs to the null space of $\bL^+$ (generated by the eigenvector $\bm{1}$).

This mathematical fact corresponds to the assumption that the power grid has automatic redispatch/balancing mechanisms, in which the total power injection mismatch is distributed uniformly among all the nodes, thus ensuring that the total net power injection is always zero.

Denote by $J \in \R^{n \times n}$ the matrix with all entries equal to one. Exploiting the eigenspace structure of $\bL$, $\bL^+$ can be calculated as
\begin{equation*}
	\bL^+ = \Big(\bL + \frac{1}{n} \bJ\Big)^{-1} - \frac{1}{n} \bJ,
\end{equation*}

In the literature, instead of $\bL^+$ it is commonly used another matrix $\bar{\bL}$, calculated using the inverse of the $(n-1) \times (n-1)$ sub-matrix obtained from $\bL$ by means of deleting the first row and first column. In our method we are implicitly choosing an average value of zero as a reference for the nodes voltage phase angles, while in the classical one the first node is used as reference by setting is phase angle equal to zero. We remark that these two procedure are equivalent if one is interested in the line power flows, as these latter depend only on the phase angle differences. However, the matrix $\bar{\bL}$ does not account for the distributed slack, which needs to added by post-multiplying by the matrix $\bS= \bI - \frac{1}{n} \bJ \in \R^{n \times n}$.

The real line power flows $\tbf$ are related with the phase angles $\btheta$ via the linear relation $\tbf = \bB \bC \btheta$. In view of Eq.~\eqref{eq:dcapprox}, the line power flow $\tbf$ can be written as a linear transformation of the power injections $\bp$, i.e.
\begin{equation}
\label{eq:tfBLp}
	\tbf= \bB \bC \bL^+ \bp.
\end{equation}
It is convenient to look at the \textit{normalized line power flow} vector $\bfl \in \R^m$, defined component-wise as $f_\ell= \tf_\ell / C_\ell$ for every $\ell=1,\dots,m$, where $C_\ell$ is the \textit{line threshold} of line $\ell$, which is assumed to be given.
Line thresholds are in place because a protracted current overload would heat up the line, causing sag, loss of tensile strength and eventually mechanical failure. If this happens, the failure may cause a global redistribution of the line power
flows which could trigger cascading failures and blackouts.

The relation between line power flows and normalized power flows can be rewritten as $\bfl = \bW \tbf$, where $\bW$ is the $m \times m$ diagonal matrix $\bW=\mathrm{diag}(C_1^{-1}, \dots, C_m^{-1})$. In view of Eq.~\eqref{eq:tfBLp}, the normalized power flows $\bfl$ can be expressed in terms of the power injections $\bp$ as
\begin{equation*}
	\bfl= \bV \bp,
\end{equation*}
where $\bV=\bW \bB \bC \bL^+ \in\R ^{m\times n}$.

\subsection{Stochastic and deterministic injections}
We now briefly outline how the model presented above can be extended to a setting where only a subset of nodes houses stochastic power injections (modeling wind and solar parks), while the other nodes house deterministic injections (corresponding to conventional controllable power plants).

First, we introduce the following notation: if $\mathbf{z}$ is a $n$-dimensional multivariate Gaussian random vector with mean $\bm{\lambda}$ and covariance matrix $\bm{\Lambda}$, it will be denoted by
$\mathbf{z} \sim\mathcal{N}_n(\bm{\lambda},\bm{\Lambda})$. 

Define the following:
\small
\begin{equation*}
\begin{aligned}
&n_s \; && \text{number of stochastic buses}  ,\\
&n_d\; && \text{number of deterministic buses}  ,\\
&\mathcal{I}_s\subseteq \{1,\ldots,n\}\; && \text{indices of stochastic buses}  ,\\
&\mathcal{I}_d\subseteq \{1,\ldots,n\}\; && \text{indices of deterministic buses}  ,\\
&\bp_s=(p_i)_{i\in \mathcal{I}_s}\in \R^{n_s}\; && \text{stochastic power injection}  ,\\
&\bp_d=(p_i)_{i\in \mathcal{I}_d}\in \R^{n_d}\; && \text{deterministic power injection}  ,\\
&\bV_s\in \R^{m\times n_s}\; && \text{matrix consisting of the}\\
&  && \text{columns of $\bV$ indexed by $\mathcal{I}_s$} ,\\
&\bV_d\in \R^{m\times n_d}\; && \text{matrix consisting of the} \\
&  && \text{columns of $\bV$ indexed by $\mathcal{I}_d$} ,\\
&\bfl_s=\bV_s\bp_s\in \R^{m}\; && \text{stochastic component of $f$}  ,\\
&\bfl_d=\bV_d\bp_d \in \R^{m}\; && \text{deterministic component of $f$}.\\
\end{aligned}
\end{equation*}
\normalsize

If a bus hosts both stochastic and deterministic generators, it is considered a stochastic bus. Stochastic power injections are modelled by mean of a $n_s$-dimensional multivariate Gaussian random vector with mean $\bmu_s\in\R^{n_s}$ and covariance matrix $\Sp \in\R^{n_s\times n_s}$, which we denote by
\[
	\bp_s\sim \mathcal{N}_{n_s}(\bmu_s,\eps\Sp ),
\]
With the previous notation, the normalized power flows can be decomposed as $\bfl=\bfl_s+\bfl_d=\bV_s \bp_s+\bfl_d$, where
\begin{align}
&\bfl_s\sim \mathcal{N}_{m}(\bnu_s,\eps\Sf), \nonumber \\
&\bnu_s=\bV_s\bmu_s, \nonumber\\
&\Sf=\bV_s\Sp  \bV_s^\top. \label{eq:Sf}
\end{align}
The nominal power flows values are thus equal to  $\bnu=\bnu_s+\bfl_d$. The decay rate for an overload in line $\ell$, analogously to formula (6) in the Main Body of the paper, is given by
\begin{equation*}
I_{\ell}=\inf_{\bp_s\in\mathbb{R}^{n_s} \,:\, |\bhe_\ell^\top (\bV_s \bp_s+\bfl_d)|\ge 1} \frac{1}{2} (\bp_s-\bmu_s)^\top \Sp ^{-1} (\bp_s-\bmu_s).
\end{equation*}
Provided that $\nu_\ell \neq 0$, the solution is unique and reads
\begin{equation}\label{eq:dec_1}
	\bp_s^{(\ell)}	= \frac{(\sign(\nu_\ell)-\nu_{\ell})}{\sigma_{\ell}^2} \Sp  \bV_s^\top \bhe_{\ell} +\bmu_s  \in \R^{n_s},
\end{equation}
where $\sigma_{\ell}^2=(\Sf)_{\ell,\ell}$. The corresponding most likely realization for power flows reads
\begin{align}\label{eq:dec_2}
	\bfell &= \bV_s\bpell_s+\bfl_d \nonumber\\
	&= \frac{(\sign(\nu_\ell)-\nu_{\ell})}{\sigma_{\ell}^2} \bV_s \Sp  \bV_s^\top \bhe_{\ell}+\bnu_s+\bfl_d\in\R^m.
\end{align}
In the next section we prove these claims for the particular case of $n_s=n$.

\section{Large deviations principles for failure events}
\subsection{Gaussian case}
In this section we provide proofs for Eqs.~$(3)$-$(7)$ in the Main Body. For the sake of clarity we present here only the proofs for the case $n=n_s$, and we remark that Eqs.~\eqref{eq:dec_1}-\eqref{eq:dec_2} in the Supplemental Material can be proved along similar lines.
In the following, we write $\bp_\eps$ and $\bfl_{\eps}$ to stress the dependence of the power injections and of the line power flows on the noise parameter $\eps$.

\begin{customthm}{1}\label{prop:first_failure}
Assume that $\max_{j=1,\ldots,m} |\nu_j| <1$. Then, for every $\ell=1,\ldots,m$, the sequence of line power flows $(\bfl_{\eps})_{\eps >0}$ satisfies the large deviations principle
\begin{equation}\label{eq:LD}
\lim_{\eps\to 0} \eps \log \P ( |(\bfl_{\eps})_\ell| \ge 1 )=- \frac{(1-|\nu_\ell|)^2}{2\sigma_\ell^2}=-I_{\ell}.
\end{equation}
The most likely power injection configuration $\bpell \in \R^n$ given the event $|(\bfl_{\eps})_\ell| \ge 1$ is the solution of the variational problem
\begin{equation}
\label{eq:varpro}
	\bpell = \arginf_{\bp\in\mathbb{R}^n \,:\, |\bhe_\ell^\top \bV \bp|\ge 1} \frac{1}{2} (\bp-\bmu)^\top \Sp ^{-1} (\bp-\bmu),
\end{equation}
which, when $\nu_\ell \neq 0$, can be explicitly computed as
\[
        \bpell = \bmu + \frac{(\sign(\nu_\ell)-\nu_{\ell})}{\sigma_{\ell}^2} \Sp  \bV^\top \bhe_{\ell}.
\]
\end{customthm}

The next proposition shows that the conditional distribution of $\bp_{\eps}$, given $|(\bfl_{\eps})_{\ell}|\ge 1$, gets concentrated around $\bpell$ exponentially fast as $\eps \to 0$, motivating the interpretation of $\bpell$ as the most likely power injection configuration given the failure of line $\ell$.
\begin{customthm}{2}\label{prop:mostlikely_injection}
Assume that $\max_{k=1,\ldots,m}|\nu_k|<1$, and that $\nu_{\ell}\neq 0$. Then, for all nodes $i=1,\ldots,n$, and for all $\delta>0$,
	\begin{equation*}
		\lim_{\eps\to 0}\eps \log \mathbb{P}(( \bp_{\eps})_i\notin (p_i^{(\ell)}-\delta,p_i^{(\ell)}+\delta) \, \big\rvert \, |(\bfl_{\eps})_\ell|\ge 1)<0.
	\end{equation*}
\end{customthm}

The line power flows corresponding to the power injection configuration $\bpell$ can be calculated as
\[
        \bfell = \bV \bpell = \bnu + \frac{(\sign(\nu_\ell)-\nu_{\ell})}{\sigma_{\ell}^2} \bV \Sp  \bV^\top \bhe_{\ell} \in \R^m.
\]
 We observe that the vectors $\bpell$ and $\bfell$ are equal to the conditional expectation of the power injections $\bp_\eps$ and power flows $\bfl_{\eps}$, respectively, conditional on the failure event $f_\ell=\sign(\nu_\ell)$, namely
\begin{align}
 &       \bpell=\E[\bp_\eps  \,\vert\, (\bfl_{\eps})_\ell=\sign(\nu_\ell)],\label{eq:realization_V2}\\
&		 \bfell=\E[\bfl_{\eps}  \,\vert\, (\bfl_{\eps})_\ell=\sign(\nu_\ell)] \nonumber.
\end{align}
In particular, for every $k=1,\dots,m$,
\[
        f^{(\ell)}_k =\nu_{k}+(\sign(\nu_{\ell})-\nu_{\ell})\frac{\Cov(f_{\ell},f_{k})}{\Var(f_{\ell})}.
\]
Note that the case $\nu_\ell=0$ has been excluded only for compactness. Indeed, in this special case the variational problem~\eqref{eq:varpro} has two solutions, $\bp^{(\ell,+)}$ and $\bp^{(\ell,-)}$. This can be easily explained by observing that if the power flow on line $\ell$ has mean $\nu_\ell=0$, then it is equally likely for the overload event $\{|f_\ell| \ge 1\}$ to occur as $\{f_\ell \ge 1\}$ or as $\{f_\ell \le -1\}$ and the most likely power injection configurations that trigger them can be different.

The previous proposition immediately yields the large deviations principle also for the first line failure event $\|\bfl_{\eps}\|_{\infty}\ge 1$, which reads
\[
        \lim_{\eps\to 0} \eps \log \P ( ||\bfl_{\eps}||_{\infty} \ge 1 ) =        -\min_{\ell=1,\ldots,m } \frac{(1-|\nu_\ell|)^2}{2\sigma_\ell^2}.
\]
Indeed, the decay rate for the event that at least one line fails is equal to the minimum of the decay rates for the failure of each line. The most likely power injections configuration that leads to the event $\|\bfl_{\eps}\|_{\infty}\ge 1$ is $\bp^{(\ell^*)}$ with $\ell^*=\argmin_{\ell=1,\ldots,m} \frac{(1-|\nu_\ell|)^2}{2\sigma_\ell^2}$.\\

\textit{Proof of Proposition~\ref{prop:first_failure}.}
Let $(\bZ^{(i)})_{i\in\mathbb{N}}$ be a sequence of i.i.d.~ $m$-dimensional multivariate normal vectors $\bZ^{(i)}\sim\N_m ( \bnu, \Sf)$, and let $\bS_k=\frac{1}{k}\sum_{i=1}^k Z^{(i)}$ be the sequence of the partial sums. By setting $\eps=\frac{1}{k}$, it immediately follows that that $\smash{\bfl_{\eps}\ed S_k}$, where $\ed$ denotes equality in distribution. Denote $g(\bp)=\frac{1}{2} (\bp-\bmu)^\top \Sp^{-1} (\bp-\bmu)$. Following~\cite[Section 3.D]{Touchette2009}, we get
\begin{align}
        & \lim_{\eps\to 0} \eps \log \P ( (\bfl_{\eps})_\ell \ge 1 )
        =\lim_{k\to \infty} \frac{1}{k} \log \P ( (\bS_k)_\ell \ge 1 )= \nonumber\\
        &= - \inf_{\bp\in\mathbb{R}^n \,:\, \bhe_\ell^\top \bV \bp \ge 1} g(\bp)
        = - \frac{(1-\nu_\ell)^2}{2\sigma_\ell^2},\label{eq:LDP_linek+}\\
        & \lim_{\eps\to 0} \eps \log \P ( (\bfl_{\eps})_\ell \le -1 )
        =\lim_{k\to \infty} \frac{1}{k} \log \P ( (\bS_k)_\ell \le -1 )= \nonumber\\
        &= - \inf_{\bp\in\mathbb{R}^n \,:\, \bhe_\ell^\top \bV \bp \le -1} g(\bp)
        = - \frac{(-1-\nu_\ell)^2}{2\sigma_\ell^2}.\label{eq:LDP_linek-}
\end{align}
The optimizers of problems \eqref{eq:LDP_linek+} and \eqref{eq:LDP_linek-} are easily computed respectively as as
\begin{align*}
	& \bp^{(\ell,+)} =  \bmu + \frac{(1-\nu_{\ell})}{\sigma_{\ell}^2} \Sp \bV^\top \bhe_{\ell},\\
	& \bp^{(\ell,-)} =  \bmu + \frac{(-1-\nu_{\ell})}{\sigma_{\ell}^2} \Sp \bV^\top \bhe_{\ell}.
\end{align*}

Note that trivially
\begin{align*}
	& \inf_{\bp\in\mathbb{R}^n \,:\, |\bhe_\ell^\top \bV \bp| \ge 1} g(\bp) =\\
	& \qquad =\min\Big \{ \inf_{\bp\in\mathbb{R}^n \,:\, \bhe_\ell^\top \bV \bp \ge 1} g(\bp), \, \inf_{\bp\in\mathbb{R}^n \,:\, \bhe_\ell^\top \bV \bp \le -1} g(\bp)\Big \},
\end{align*}
and thus identities \eqref{eq:LD} and \eqref{eq:varpro} immediately follow. \qed\\

\textit{Proof of Proposition~\ref{prop:mostlikely_injection}.} We have
 \begin{align*}
 &\log \,\mathbb{P}(( \bp_{\eps})_i\notin (p_i^{(\ell)}-\delta,p_i^{(\ell)}+\delta) \, \big\rvert \, |(\bfl_{\eps})_\ell|\ge 1)\\
=&\log\,  \mathbb{P}(( \bp_{\eps})_i\notin (p_i^{(\ell)}-\delta,p_i^{(\ell)}+\delta), \, | (\bfl_{\eps})_\ell|\ge 1)\\
&-\log\P ( |(\bfl_{\eps})_\ell|\ge 1).
 \end{align*}
%
Denote $g(\bp)=\frac{1}{2} (\bp-\bmu)^\top \Sp ^{-1} (\bp-\bmu)$. From large deviations theory, it holds that that
\begin{align}
	&\lim_{\eps\to 0} \eps \log \P( |(\bfl_{\eps})_\ell| \ge 1)=
	 - \inf_{\bp\in\mathbb{R}^n \,:\, |\bhe_{\ell}^\top \bV p| \ge 1} g(p) \label{eq:I_ki_power}\\
	&\lim_{\eps\to 0} \eps \log \P( (\bp_\eps)_i \notin  (p_i^{(\ell)}-\delta,p_i^{(\ell)}+\delta), \, |(\bfl_{\eps})_\ell| \ge 1)= \nonumber \\
	& \quad =-\inf_{\substack{\bp\in\mathbb{R}^n \,:\, |\bhe_{\ell}^\top \bV p| \ge 1,\\ \qquad\quad  |p_i-p^{(\ell)}_i|\ge\delta}} g(p).\label{eq:tildeI_ki_power}
\end{align}
Define the corresponding decay rates as
\begin{equation*}
	I_{\ell}=\inf_{\bp\in\mathbb{R}^n \,:\, |\bhe_\ell^\top \bV \bp| \ge 1} g(\bp), \quad	
	J_{\ell}=\inf_{\substack{\bp\in\mathbb{R}^n \,:\, |\bhe_\ell^\top \bV \bp| \ge 1,\\ \qquad\quad |p_i-p^{(\ell)}_i|\ge\delta}} g(\bp).
\end{equation*}
Then we can rewrite
\begin{equation*}
	\lim_{\eps\to 0}\eps \log \mathbb{P}(( \bp_{\eps})_k\notin (p_i^{(\ell)}-\delta,p_i^{(\ell)}+\delta) \, \big\rvert \, |(\bfl_{\eps})_\ell|\ge 1)=-J_{\ell}+I_{\ell},
\end{equation*}
and, therefore, the claim is equivalent to proving that $J_{\ell}>I_{\ell}$.
Notice that the feasible set of the minimization problem \eqref{eq:tildeI_ki_power} is strictly contained in that of the problem \eqref{eq:I_ki_power}, implying that $J_\ell\ge I_{\ell}$.

Recall that $\bpell$ is the unique optimal solution of \eqref{eq:I_ki_power}, and let $\hat{\bp}^{(\ell)}$ be an optimal solution of \eqref{eq:tildeI_ki_power}. Clearly $\hat{p}^{(\ell)}$ is feasible also for problem \eqref{eq:I_ki_power}. If it was the case that $J_{\ell}=I_{\ell}$, then $\hat{\bp}^{(\ell)}$ would be an optimal solution for \eqref{eq:I_ki_power}, and thus by uniqueness ($g(p)$ is strictly convex) $\hat{\bp}^{(\ell)}=\bpell$. But this leads to a contradiction, since $\hat{\bp}^{(\ell)}$ is by construction such that $|\hat{p}_i-\hat{p}^{(\ell)}_i|\ge\delta$.
Hence $J_\ell > I_\ell$ and we conclude that
\begin{equation*}
		\lim_{\eps\to 0}\eps \log \mathbb{P}(( \bp_{\eps})_i\notin (p_i^{(\ell)}-\delta,p_i^{(\ell)}+\delta) \, \big\rvert \, |(\bfl_{\eps})_\ell|\ge 1)<0. \qed
\end{equation*}

\subsection{Extension to non-Gaussian case}
In this section we briefly describe how to extend the analyis to the non-Gaussian scenario. Consider a model for the power injection vector given by
\begin{equation*}
\bp_{\eps}=\bmu+\sqrt{\eps} \bX,
\end{equation*}
where $\bmu\in\R^n$ and $\bX=(X_1,\ldots,X_n)$ is a random vector with mean $0$ and log-moment generating function
\[\log M(\bs)=\log\E[e^{\langle \bs,\bX \rangle}].\]
The power flows vector is thus given by
$\bfl_{\eps}=\bV \, \bp_{\eps}$.
Define the Fenchel-Legendre (also known as the convex conjugate) transform of $\log M(\bs)$, i.e.
\begin{equation*}
\Lambda^*(\bx)=\sup_{s\in\mathbb{R}^n}(\langle \bs,\bx \rangle-\log M(\bs)).
\end{equation*}
Then, for every $\ell=1,\ldots,m$, the sequence $(\bfl_{\eps})_{\eps>0}$ satisfies the large deviations principle (see~\cite{DemboZeitouni2009})
\begin{equation*}
\lim_{\eps\to 0} \eps \log \P ( |(\bfl_{\eps})_\ell| \ge 1 )=
- \inf_{\bx\in\mathbb{R}^n \,:\,|\bhe_{\ell}^\top \bV (\bmu+\bx)| \ge 1} \Lambda^*(x),
\end{equation*}
and the most likely power injection configuration $\bpell \in \R^n$ given the event $|(\bfl_{\eps})_\ell| \ge 1$ is
\begin{equation*}
	\bpell = \bmu+\arginf_{\bx\in\mathbb{R}^n \,:\, |\bhe_{\ell}^\top \bV (\bmu+\bx)|\ge 1}\Lambda^*(x).
\end{equation*}
The rest of the analysis can then be carried out along similar lines as we did for the Gaussian case.

\section{Power flow redistribution}
For every line $\ell$ define $\mathcal{J}(\ell)$ to be the collection of lines that fail jointly with $\ell$ as
\[
	 \mathcal{J}(\ell) =\{ k ~:~ |f^{(\ell)}_k| \ge 1\}.
\]
Let $j(\ell)=|\mathcal{J}(\ell)|$ be its cardinality and note that $j(\ell) \ge 1$ as trivially $\ell$ always belongs to $\mathcal{J}(\ell)$. Denote by $\widetilde{G}^{(\ell)}$ the graph obtained from $G$ by removing all the lines in $\mathcal{J}(\ell)$.

Let us focus first on the case of the isolated failure of line $\ell$, that is when $\mathcal{J}(\ell)=\{\ell\}$. In this case $\widetilde{G}^{(\ell)}=G(V,E\setminus \{\ell\})$ is the graph obtained from $G$ after removing the line $\ell=(i,j)$. Provided that the power injections remain unchanged, the power flows redistribute among the remaining lines. Using the concept of \textit{effective resistance matrix} $\bR \in\R^{n\times n}$ and under the DC approximation, in~\cite{CKvM16,Schaub2014,Soltan2015} it is proven that alternative paths for the power to flow from node $i$ to $j$ exist (i.e., $\widetilde{G}^{(\ell)}$ is still connected) if and only if $\beta_{i,j} R_{i,j} \neq 1$. In other words, $\beta_{i,j} R_{i,j}=1$ can only occur in the scenario where line $\ell=(i,j)$ is a \textit{bridge}, i.e., its removal results in the disconnection of the original graph $G$ in two components. If $\widetilde{G}^{(\ell)}$ is still a connected graph, the power flows after redistribution $\frv \in \R^{m-1}$ are related with the original line flows $\bfl\in \R^m$ in the network $G$ by the relation
\[
	\bar{f}^{(\ell)}_k=f_k+  f^{(\ell)}_\ell \phi^{(\ell)}_{k}, \quad \text{ for every } k \neq \ell,
\]
where $f^{(\ell)}_\ell = \pm 1$ depending on the way the power flow on line $\ell$ exceeded the threshold $1$. If $\ell=(i,j)$ and $k=(a,b)$ the coefficient $\phi^{(\ell)}_{k} \in \R$ can be computed as
\begin{equation}\label{eq:LODF_def}
	\phi^{(\ell)}_{k}=\phi_{(i,j),(a,b)} = \beta_k \cdot \frac{C_\ell}{C_k} \cdot \frac{R_{a,j}-R_{a,i}+R_{b,i}-R_{b,j}}  {2(1-\beta_\ell R_{i,j})},
\end{equation}
The ratio $C_\ell/C_k$ appears in the latter formula since we work with normalized line power flows and we correspondingly defined $\bm{\phi}^{(\ell)}=\{\phi^{(\ell)}_k\}_{k\neq \ell}$ to be the normalized version of the classical line outage distribution factors (LODF, \cite{Guo2009}).
Moreover, we define the \textit{most likely} power flows configuration $\bfmostlikely\in\R^{m-1}$ after redistribution as

\begin{equation}\label{eq:mostlikely_isolated}
	\fmostlikely_k=f^{(\ell)}_k+ f^{(\ell)}_\ell\phi^{(\ell)}_{k} , \quad \text{ for every } k \neq \ell.
\end{equation}

\subsection{Ring topology} 
We now focus on a particular topology, namely the ring on $n$ nodes, which we use as an illustrative example to show the non-locality of cascades in the Main Body. In this topology, nodes are placed on a ring and each node is connected to its previous and subsequent neighbor. Denote the set of nodes as $\mathcal{N}=\{1,\ldots,m\}$ and the set of lines $\mathcal{L}=\{l_1,\ldots,l_n\}$, where $l_1=(n,1),l_2=(1,2),\ldots,l_n=(n-1,n)$. It is easy to prove that, in a ring network with homogeneous line thresholds and unitary susceptances, $\phi_{\ell,k}=-1$ for every $\ell \neq k$.

\begin{lem}
Consider a ring network with homogeneous line thresholds ($C_\ell=C$ for every line $\ell$) and homogeneous unitary susceptances ($\beta_{\ell}=1$ for every line $\ell$). Then
\begin{enumerate}
\item[i)] The effective resistance between a pair of nodes $i,j$ is given by
\begin{equation}\label{eq:R_circuit}
R_{i,j}=\frac{|j-i|(n-|j-i|)}{n}.
\end{equation}
\item[ii)] For every pair of lines $l_{k}=(k-1,k),l_{\ell}=(\ell-1,\ell)$, with $k\neq \ell$, the LODF is constant and equal to
\[
	\phi_{(k-1,k),(\ell-1,\ell)}=-1.
\]
\end{enumerate}
\end{lem}
\begin{proof}
i) See identity $(4)$ in~\cite{Bapat2006}.
ii) First, observe that the effective resistance between two adjacent nodes $i$ and $j=i+1$ in a circuit graph is equal to $R_{i,j}=\frac{n-1}{n}$, thanks to Eq.~\eqref{eq:R_circuit}. After a straightforward calculation, and using that $\beta_{\ell}=1, C_\ell=C$ for all lines $\ell$, Eq.~\eqref{eq:LODF_def} becomes
\[
	\phi_{(k-1,k),(\ell-1,\ell)}= -\frac{2}{2n\Bigl(1-\frac{n-1}{n}\Bigr)}=-1.\qedhere
\]
\end{proof}

\subsection{General topology} 
Going back to the case of a general network topology and any type of failures, isolated or joint, the power flows after redistribution $\frv \in \R^{m-j(\ell)}$ are related with the power injections $\bp \in \R^n$ by the relation
\[
	\frv = \widetilde{\bV}^{(\ell)} \bp,
\]
where the  $(m-j(\ell)) \times n$ matrix $\widetilde{\bV}^{(\ell)}$ can be constructed analogously to $\bV$, but considering the altered graph $\widetilde{G}^{(\ell)}$ instead of $G$.
We define the \textit{most likely} power flow configuration $\bfmostlikely$ after redistribution as
\[
	\bfmostlikely = \widetilde{\bV}^{(\ell)} \bpell,
\]
which generalizes Eq.~\eqref{eq:mostlikely_isolated} to any kind of failure, isolated of joint.
The next proposition shows that it is enough to look at the vector $\bfmostlikely$ to determine whether a line that survived at the first cascade stage (i.e., that did not fail jointly with $\ell$) will fail with high probability or not after the power redistribution (i.e., at the second cascade stage).

\begin{customthm}{3}\label{prop:second_failure}
Assume that $\max_{k=1,\ldots,m}|\nu_k|<1$, and that $\nu_{\ell}\neq 0$. Then, for all lines $k\in E\setminus \mathcal{J}(\ell)$, and for all $\delta>0$,
\[
		\lim_{\eps\to 0}\eps \log \mathbb{P}((\frv_{\eps})_k\notin (\fmostlikely_k-\delta,\fmostlikely_k+\delta) \, \big\rvert \, |(\bfl_{\eps})_\ell|\ge 1)<0.
\]
In particular, if $|\fmostlikely_k|\ge 1$, then
\[
		 \mathbb{P}(|(\frv_{\eps})_k|\ge 1\, \big\rvert \, |(\bfl_{\eps})_\ell|\ge 1)\to 1\quad\text{ as }\eps\to 0,
\]
exponentially fast in $1/\eps$.
\end{customthm}

\textit{Proof.} 
Let $A_{\eps,k}$ denote the event $A_{\eps,k}=\{(\frv_{\eps})_k\notin (\fmostlikely_k-\delta,\fmostlikely_k+\delta)\}$, and define $Q=\lim_{\eps\to 0}\eps \log \mathbb{P}(A_{\eps,k}\, \big\rvert \, |(\bfl_{\eps})_\ell|\ge 1)$.
The proof that $Q<0$ is analogous to the proof of Prop.~\ref{prop:mostlikely_injection}.
For the second part, it follows from $\lim_{\eps\to 0}\eps \log \mathbb{P}(A_{\eps,k}\, \big\rvert \, |(\bfl_{\eps})_\ell|\ge 1)=Q$ that for every $\eta>0$ there exists a $\bar{\eps}$ such that, for every $\eps <\bar{\eps}$,
\[Q-\eta\le\eps \log \mathbb{P}(A_{\eps,k}\, \big\rvert \, |(\bfl_{\eps})_\ell|\ge 1) \le Q+\eta,\]
and thus
\[\exp\Bigl(\frac{Q-\eta}{\eps}\Bigr)\le \mathbb{P}(A_{\eps,k}\, \big\rvert \, |(\bfl_{\eps})_\ell|\ge 1)\le\exp\Bigl(\frac{Q+\eta}{\eps}\Bigr).\]
Let $A_{\eps,k}^c$ denote the complementary event of $A_{\eps,k}$.
Since $|\fmostlikely_k|\ge 1$, for $\delta$ sufficiently small we have
\[
A_{\eps,k}^c=\{|(\frv_{\eps})_k-\fmostlikely_k|<\delta\}\subseteq\{|(\frv_{\eps})_k|\ge 1\},\]
yielding
\begin{align*}\mathbb{P}(|(\frv_{\eps})_k|\ge 1\, \big\rvert \, |(\bfl_{\eps})_\ell|\ge 1)\ge \,
&\mathbb{P}(A_{\eps,k}^c \,\big\rvert \, |(\bfl_{\eps})_\ell|\ge 1)\\
 \ge\, & 1-\exp\Bigl(\frac{Q+\eta}{\eps}\Bigr).
\end{align*}
Since $Q<0$, the result follows. \qed
\section{Application: SciGRID German network}
We now demonstrate our methodology in the case of a real-world power grid and a realistic system state.
\subsection{Dataset Description}
We perform our experiments using PyPSA, a free software toolbox for power system analysis~\cite{PyPSA2017}. We use the dataset described in \cite{SCIGRID1,SCIGRID2}, which provides a model of the German electricity system based on SciGRID and OpenStreetMap~\cite{SCIGRID00,OpenStreetMap2017}.

The dataset includes load/generation time series and geographical locations of the nodes, differentiating between renewable and conventional generation. It also provides data for transmission lines limits, transformers, generation capacity and marginal costs, allowing us to couple our theoretical analysis with realistic Optimal Power Flow (OPF, \cite{Huneault1991}) computations. The time-series provide hourly data for the entire year $2011$. For more technical information on the dataset, we refer to \cite{SCIGRID1,SCIGRID2}.

The SciGRID German network consists of $585$ buses, $1423$ generators including conventional power plants and wind and solar parks, $38$ pump storage units, $852$ lines and $96$ transformers. For the analysis carried out in this paper, storage units are not included and we exclude transformer failures.
The renewable generators are divided in three classes, solar, wind onshore and wind offshore. Each bus can house multiple generators, both renewable and conventional, but it is limited to at most one renewable generator for each class. Let $\N_{\text{w.off}},\N_{\text{w.on}},\N_{\text{sol}}$ denote the set of buses housing, respectively, wind offshore, wind onshore and solar generators, with $|\N_{\text{w.off}}|=5,\,|N_{\text{w.on}}|=488,\,|N_{\text{sol}}|=489$, and $\N_{\text{w.off}}\subseteq \N_{\text{w.on}}\subseteq \N_{\text{sol}}$. The remaining $96$ buses house $441$ conventional generators.

Let $n_s=489$ denote the total number of buses housing renewable generators. If a bus houses both renewable and conventional generators, it will be considered a \textit{stochastic} bus for our decomposition formulation.
We model stochastic net power injections by means of a multivariate Gaussian random vector $\bp_s\sim \mathcal{N}_{n_s}(\bmu_s,\eps\Sp )$. 

The distinction between the noise parameter $\eps$ and the covariance matrix $\Sp$ is relevant only for the theoretical analysis (where we take the limit $\eps \to 0$ while the matrix $\Sigma_p$ is fixed), since as far as the numerical case study is concerned, all the results are obtained by using the product $\eps\Sp$, which is directly estimated from the SciGRID data. In the following, we will thus take $\eps=1$ and refer to the covariance matrix of $\bp_s$ simply as $\Sp$. 

\subsection{Data-based model for $\mathbf{\mu}_s$}
In order to get a realistic nominal line flows value $\bnu$, we perform a linear OPF relative to the day $01/01/2011$, for different hours of the day. A linear OPF consists of minimizing the total cost of generation, subject to energy balance, generation and transmission lines constraints, under the assumptions of the DC approximations.
In order to model a heavily-loaded but not overloaded system, in the OPF we scale the true line limits $C_{\ell}$ by a contingency factor of $\lambda=0.7$. This is a common practice in power engineering that allows room for reactive power flows and stability reserve.

More precisely, let $g(i)$ be the generation at bus $i$ as outputted by the OPF for a given hour, and let us write it as $g(i)=g_r(i)+g_d(i)$, with $g_r(i)$ the power produced by renewable generators attached to the bus, and $g_d(i)$ the power supplied by conventional generators. If the demand at bus $i$ is given by $d(i)$, then the average stochastic power injection vector $\bmu_s\in\R^{n_s}$ is modeled as
\[
({\bmu_s})_{i}=g_r(i)+g_d(i)-d(i),\quad i=1,\ldots,n_s,
\]
while the deterministic power injection reads $p_i=g_d(i)-d(i)$ for $i\in\mathcal{I}_d$.

\subsection{Data-based model for $\mathbf{\Sigma}_p$}
In order to model the fluctuations of renewable generation around the nominal values, and thus estimate $\Sp $, we use realistic hourly values of wind and solar energy production to fit a stochastic model. We then use the steady-state covariance of the model residuals as an estimate for $\Sp $.
Following~\cite{Milligan2003}, we choose to use AutoRegressive-Moving-Average (ARMA) models, which we describe in details below.

Note that we do not aim to find the best possible stochastic model for renewable generation, which is beyond the scope of this paper, but instead to provide an estimate for the covariance matrix $\Sp $ in order to validate our theoretical results, which are asymptotically valid in a small-noise regime. We speculate that more sophisticated models, and/or data on smaller time-scales, may lead to smaller values for the correlations in $\Sp$, thus getting closer to the small noise limit. 

We now describe the estimation procedure for $\Sp $ (as mentioned before we normalize $\eps=1$ in our empirical study). 
The SciGRID dataset contains time series
\[\by_{\text{w.off}}\in\R^{M\times 5},\by_{\text{w.on}}\in\R^{M\times 488},\by_{\text{sol}}\in\R^{M\times 489},\]
for the available power output of wind offshore, wind onshore and solar generators, for each hour of the year $2011$, accounting for a total of $M=8760$ measurements for each generator~\cite{SCIGRID1}.
For each time series, $y_{(\cdot)}(t,j)$ denote the available power output at time $t$ for the $j$-th generator of a given type, in MW units.

%
%
%

\subsubsection{Wind power model}
As a pre-processing step, we merge together the two time series $\by_{\text{w.off}},\by_{\text{w.on}}$ by summing up the onshore and offshore wind power at the buses $\N_{\text{w.off}}\subseteq \N_{\text{w.on}}$. This yields the time series of wind power production
\[y_\text{w}(t,j)=y_{\text{w.on}}(t,j)+\mathds{1}_{j\in \N_{\text{w.off}}}y_{\text{w.off}}(t,j),\]
where $\mathds{1}_{\{\}}$ is the indicator function of the event in the bracket, taking value $1$ if the event is satisfied, and $0$, otherwise.

We select one portion of the data $\{1,\ldots,T\}\subseteq \{1,\ldots,M\}$, corresponding to the month of January, to be used to fit the model. For each windpark $j$, following~\cite{Milligan2003} we consider an ARMA(1,24) model of the form
\begin{align*}
x(t,j) =& a_{1,j}x(t-1,j)+e(t,j)\\
+& m_{1,j}e(t-1,j)+\ldots+m_{24,j} e(t-24,j),
\end{align*}
where $x(t-1,j)$ is the auto-regressive term, and $e(t-k,j)$, $k=1,\ldots,24$, are the white-noise error terms. For each windpark $j$, we fit the above model to the wind power data  $\{y_\text{w}(t,j)\}_{t=1:T}$ in R using the function \textit{arima},
and consider the time series of the residuals $e_\text{w}(1,j),\ldots,e_\text{w}(T,j)$.

The empirical variance of the residuals is used as proxy for the variance of the output of windpark $j$, namely
\[
	(\bSigma_\text{w})_{jj}=\widehat{\Var}(e_\text{w}(1,j),\ldots,e_\text{w}(T,j)),
\]
where $\widehat{\Var}$ denotes the empirical variance. In a similar way, the empirical covariance of the residuals is used to model the covariance between the output of windparks $i$ and $j$, $i\neq j$, namely
\[
	(\bSigma_\text{w})_{ij}=\widehat{\Cov}\Bigl(\{{e_\text{w}(t,i)}\}_{t=1:T},\{{e_\text{w}(t,j)}\}_{t=1:T}\Bigr),
\]
where $\widehat{\Cov}$ denotes the empirical covariance.
\subsubsection{Solar power model}
State-of-the-art models for solar irradiance often combine statistical techniques with cloud motion analysis and numerical weather prediction (NWP) models, see~\cite{Antonanzas2016} for a review. 
Since the available data in our case study are limited to historical records for power production of solar generators, and do not include any weather data, we used the purely statistical model ARMA($p$,$q$), which has been used succesfully in~\cite{Huang2012}.

Regarding the orders $p,q$ of the ARMA model, after some exploratory analysis we decided to use an ARMA(24,24) model with all parameters fixed to $0$, except for the ones corresponding to the seven hours before, and the one corresponding to twenty-four hours before. The rationale behind this choice is that by using the value corresponding to twenty-four hours before, we capture the dependency on the hour of the day, while the values from $7$ hours before capture the shape of the current day. 

More precisely, the model reads
\begin{align*}
x(t,j)& =a_{1,j}x(t-1,j)+\ldots+a_{7,j}x(t-7,j)\\
       &\quad +a_{24,j}x(t-24,j)\\
       &\quad +e(t,j)+m_{1,j}e(t-1,j)+\ldots+ m_{7,j} e(t-7,j)\\
       &\quad+m_{24,j} e(t-24,j).
\end{align*}
For each solar park $j$, we fit the above model to the solar power data
$(y_{\text{sol}}(t,j))_{t=1:T}$,
using again the R function \textit{arima},
and consider the time series of the residuals
$(e_{\text{sol}(t,j)})_{t\in\mathcal{D}},$
where $\mathcal{D} \subseteq \{1,\ldots, T\}$ denotes the set of daylight hours of January $2011$.
The covariance matrix for the solar power generation is obtained as
\[{(\bSigma_{\text{sol}})}_{ij}=\widehat{\Cov}\Bigl({(e_{\text{sol}}(t,i))}_{t\in\mathcal{D}},{(e_{\text{sol}}(t,j))}_{t\in\mathcal{D}}\Bigr).\]

Since we perform numerical experiments for different hours of the day $01/01/2011$, we need to model renewable fluctuations taking into account whether or not we consider a daylight hour, as there is no solar energy production before sunrise and after sunset. In view of this, and assuming that the residuals for the wind and solar models are independent (see~\cite{Zhang2013corr}), we model the covariance matrix relative to an hour $h$ as
\[\Sp (h)= \bSigma_{\text{w}}+\mathds{1}_{h\in \mathcal{D}_1}\bSigma_{\text{sol}},\]
where $\mathcal{D}_1 \subseteq \{1,\ldots,24\}$ denotes the set of daylight hours of $01/01/2011$.

The magnitude of power injections noise at bus $i$ is quantified by the standard deviation $\sqrt{(\Sp)_{ii}}$, expressed as a percentage of the combined installed capacity of wind and solar generators located at bus $i$\footnote{We note that normalizing the error using the installed capacity of a generator is standard in the literature~\cite{Hodge2011}.}. In our numerical study, we find that for daylight hours the mean of these standard deviations across all buses is $8.5\%$, while during nighttime the mean reduces to $5\%$.
%

\subsection{Data-based model for $\mathbf{\Sigma}_f$}
In view of Eq.~\eqref{eq:Sf}, the covariance matrix for the line power flows $\bfl_s$ is calculated as
$\Sf=\bV_s\Sp \bV_s^\top$. The magnitude of power flows noise is quantified by the standard deviations $\sigma_\ell=\sqrt{(\Sf)_{\ell\ell}}$. Since the nominal values for the power flows $\nu_{\ell}$ have been standardized as fractions of line thresholds, and thus range within the interval $[-1,1]$, the values of $\sigma_{\ell}$ describe the magnitude of the power flows noise as a percentage of the corresponding line threshold.  
 In our numerical study, we find that for daylight hours the power flow standard deviations lie within the range $[0.00007,0.14219]$, with mean $0.0228$, while during nighttime the range is $[0.00001,0.14203]$, with mean $0.0131$.

\subsection{German network: Ranking of most vulnerable lines}
In Figs.~\ref{fig:flows}-\ref{fig:prob_11am} are reported, respectively, a heatmap for the values of normalized line flows $|\nu_{\ell}|$ and for the true failure probabilities
$\mathbb{P}(|f_{\ell}|\ge 1)$
for every transmission line in the German network, relative to the hour $11$am of the day $01/01/2011$, and for an effective line limit factor of $\lambda=0.7$.

\begin{figure}[!htb]
     \centering
\includegraphics[width=1\columnwidth]{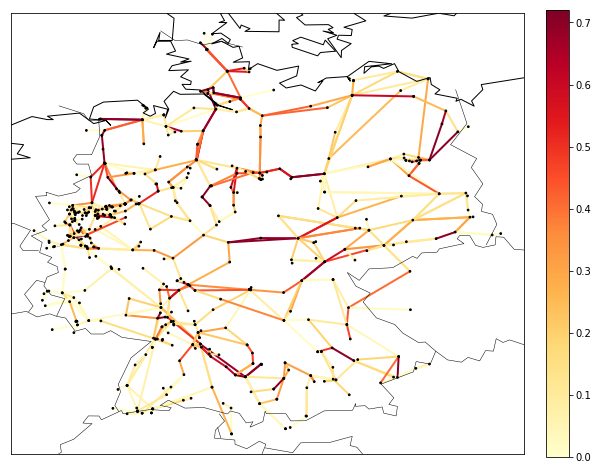}
\caption{\footnotesize Heatmap visualizing the nominal power flows values $\nu_{\ell}$ for the German network at $11$am.}\label{fig:flows}
\end{figure}

\begin{figure}[!htb]
     \centering
  \includegraphics[width=1\columnwidth]{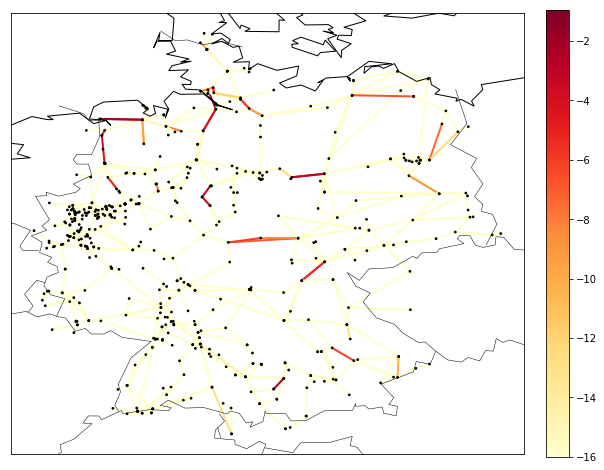}
	\caption{\footnotesize Heatmap visualizing the logarithm of the exact overload probabilities $\log_{10} \mathbb{P}(|f_{\ell}|\ge 1)$ for the German network at $11$am.}\label{fig:prob_11am}
\end{figure}



\newpage
By comparing Fig.~\ref{fig:flows} and Fig.~\ref{fig:prob_11am}, we see that a large $|\nu_{\ell}|$ does not necessarily imply a higher chance of failure, suggesting that decay rates are a better indicator of system vulnerabilities.
The most likely line to fail is line $361$, which connects two buses housing wind farms (\textit{EON Netz} and \textit{Umspannwerk Kraftwerk Emden}). This line is at capacity ($|\nu_{361}|=0.7$) and has the highest standard deviation ($\sigma_{361}=0.142$). However, we notice that a large nominal value of $|\nu_{\ell}|$ does not necessarily imply a high chance of failure.
For instance, several lines in the south of Germany have a moderate to high value $|\nu_{\ell}|$, see Fig.~\ref{fig:flows}. In particular, line $310$, which connects buses
\textit{Vöhringen Amprion} and \textit{Umspannwerk Dellmensingen},
is at capacity ($|\nu_{310}|=0.7$), but ranks only $66$-th out of $852$ lines,
with a power flow standard deviation almost one order of magnitude lower than the standard deviation of the most likely line to fail ($\sigma_{310}=0.0182$).

Fig.~\ref{fig:5ranking} depicts the $5\%$ most likely lines to fail, ranked according
to the large deviations decay rates $I_{\ell}= \frac{(1-|\nu_\ell|)^2}{2\sigma_\ell^2}$, where $\sigma_{\ell}=(\Sf)_{\ell\ell}=(\bV\Sp \bV^\top)_{\ell\ell}$. The ranking based on the large deviations approximation successfully recovers the most likely lines to fail, and, in fact, yields the same ordering as the one based on exact probabilities. As an illustration, in Table~\ref{tab:ranking} are reported the indexes, the exact failure probabilities and the decay rates for the $20$ most likely lines to fail at $11$ am.
\begin{table}[h!]
\centering
\begin{tabular}{c|c|c}
\hline
\hline
$\ell$ &  $\, \mathbb{P}(|f_{\ell}|\ge 1) \,$ &$I_{\ell}$ \\
\hline
361  &  1.743e-02 & 2.225 \\
803  &  8.228e-04 & 4.954\\
19   &  6.783e-04 & 5.132\\
27   &  6.033e-04 & 5.240\\
389  &  4.503e-04 & 5.511\\
390  &  4.460e-04 & 5.520\\
670  &  3.527e-04 & 5.737\\
809  &  7.575e-05 & 7.177\\
586 &   5.574e-05 & 7.466\\
587  &  5.454e-05 & 7.486\\
810  &  2.496e-05 &  8.225\\
712  &  6.440e-06 & 9.514 \\
682  &  5.337e-06 & 9.693\\
683  &  5.318e-06 & 9.697\\
714  &  3.876e-06 & 9.999\\
715  &  1.052e-06 &  11.249\\
554 &   4.267e-07 &  12.117\\
488  &  4.209e-07 & 12.130\\
707  &  1.199e-07 & 13.341\\
818 &   1.199e-07 &  13.341\\
\hline
\hline
\end{tabular}
\caption{Line indexes, exact failure probabilities and decay rates for the $20$ top most vulnerable lines, $11$ am.}
\label{tab:ranking}
\end{table}

\begin{figure}[t!]
    \centering
    \includegraphics[width=1\columnwidth]{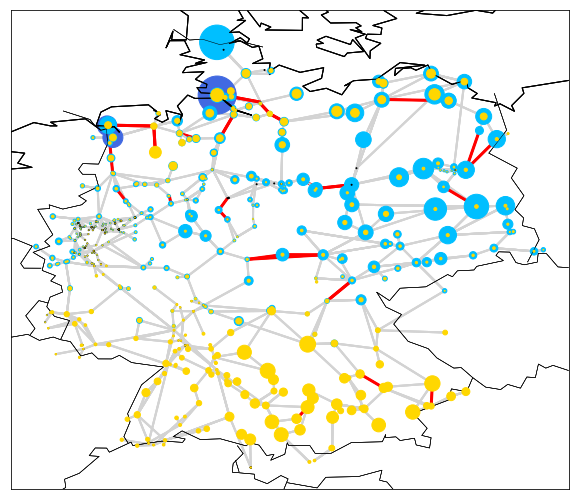}
   \caption{\footnotesize  Top $5\%$ most likely lines to fail in the German network at $11$ am. The buses housing stochastic power injections have different colors depending on the type of renewable sources (blue for wind offshore, light blue for wind onshore and yellow for solar) and sizes
proportional to the absolute values of the corresponding nominal
injections. }
        \label{fig:5ranking}
\end{figure}

The large-deviations-based ranking provides a parsimonious way to detect vulnerable lines, and can be used to appreciate qualitative differences among different hours of the day.
Table \ref{tab:percentages} lists values for total generation ($G$), and generation mix for different hours of the day ($p_{\text{w.off}}$ for wind offshore, $p_{\text{w.on}}$ for wind onshore and $p_{\text{s}}$ for solar). For example, in the morning there is more solar generation and moderate demand, while in the afternoon there is zero solar generation and higher demand.
\begin{table}[!h]
\begin{tabular}{c|c|c|c|c}
\hline
\hline
\text{hour}  & G & \, \, $p_{\text{w.off}}$ \, \, & \, \, $p_{\text{w.on}}$ \, \, & $p_{\text{sol}}$ \,\,\\
\hline
0 \text{am}& 51.75 \text{GW} & 1.7 \%& 35.6 \%& 0.0\%\\
4 \text{am}& 44.71 \text{GW} & 2.0 \%& 45.6 \%& 0.0\%\\
8 \text{am}& 44.83 \text{GW} & 4.5 \%& 44.7 \%& 8.1\%\\
11 \text{am}&  52.52  \text{GW} & 4.4 \%& 32.4 \%& 17.3\%\\
4 \text{pm} &  57.56  \text{GW} &4.1 \%& 23.9 \%& 0.0\%\\
8 \text{pm} &  54.74  \text{GW} & 4.1 \%& 22.9 \%& 0.0\%\\
\hline
\hline
\end{tabular}
\caption{\footnotesize  Total generation and renewable percentages for different hours of the day.}
\label{tab:percentages}
\end{table}

In Fig.~\ref{fig:mostlikeley_mix} the top $5\%$ most likely lines to fail are depicted (in red) for four different hours of the day, together with the nominal values outputted by the OPF for renewable generation.
By comparing Figs.~\ref{Fig:mostlikely_mix_Data1}-\ref{Fig:mostlikely_mix_Data2} and Figs.~\ref{Fig:mostlikely_mix_Data3}-\ref{Fig:mostlikely_mix_Data4}, for example, we see how solar generation is responsible for an increased number of vulnerable lines in in the south of Germany.

\begin{figure}[t!]
    \centering
    \begin{subfigure}[t]{0.48\columnwidth}
        \centering
        \includegraphics[width=1\columnwidth]{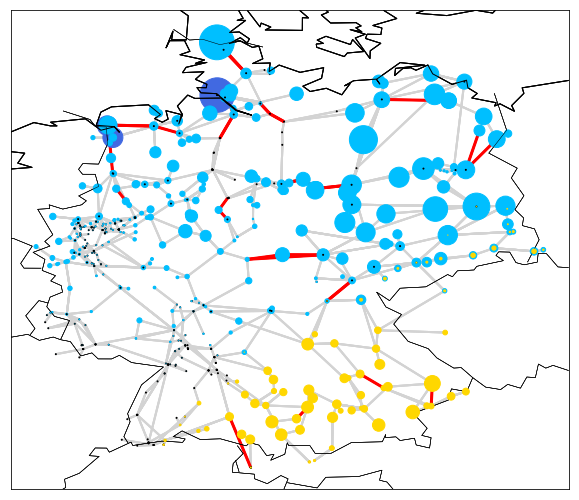}
       \caption{\footnotesize     $8$ am.}\label{Fig:mostlikely_mix_Data1}
        \label{fig:}
    \end{subfigure}%
    ~
    \begin{subfigure}[t]{0.48\columnwidth}
        \centering
        \includegraphics[width=1\columnwidth]{5percent_mix_11am_NEW.png}
 \caption{\footnotesize     $11$ am.}\label{Fig:mostlikely_mix_Data2}
    ~
      \end{subfigure}%
      \vskip\baselineskip
    \begin{subfigure}[t]{0.48\columnwidth}
        \centering
        \includegraphics[width=1\columnwidth]{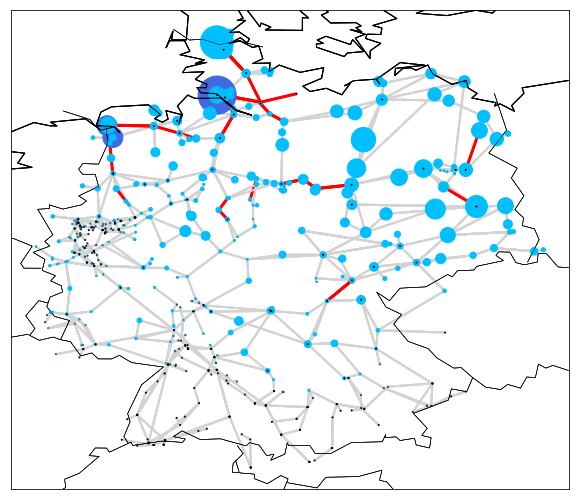}
 \caption{\footnotesize     $4$ pm.}\label{Fig:mostlikely_mix_Data3}
  \end{subfigure}
    ~
    \begin{subfigure}[t]{0.48\columnwidth}
        \centering
        \includegraphics[width=1\columnwidth]{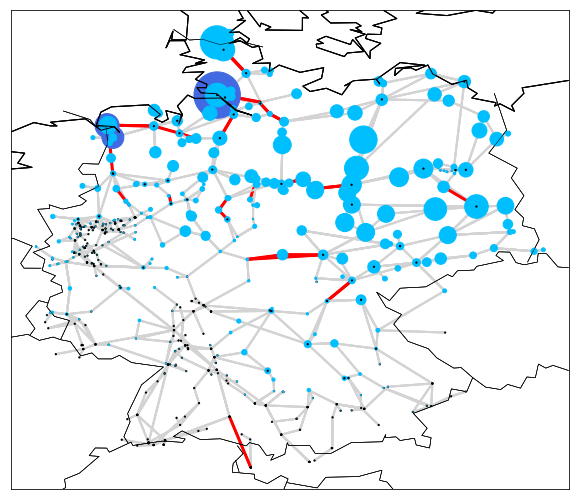}
 \caption{\footnotesize     $8$ pm.}\label{Fig:mostlikely_mix_Data4}
     \end{subfigure}
\caption{\footnotesize     Top $5\%$ most likely lines to fail (in red), together with nominal stochastic generation values.  The buses housing stochastic power injections have different colors depending on the type of renewable sources (blue for wind offshore, light blue for wind onshore and yellow for solar) and sizes
proportional to the absolute values of the corresponding nominal
injections.}
\label{fig:mostlikeley_mix}
\end{figure}

\subsection{German network: Most likely power injections}
In order to keep the notation light, in the following two subsections we omit the subscript $s$ (which refers to stochastic power injections) from the vectors $\bmu,\,\bpell,\,\bp_{\eps},\,\prv_{\eps}$.

The small-noise regime theoretical power injections configuration responsible for the failure of line $\ell$, as given by
Eq.~\eqref{eq:realization_V2}, reads $\bpell=\E[\,\bp_{\eps}\, | \,(\bfl_{\eps})_{\ell} =\sign(\nu_{\ell})]$. As an illustration, Fig.~\ref{fig:mostlikely_11am} depicts $\bpell$ leading to the isolated failure of line $720$. The bus sizes reflect how much $\bpell$ deviates from $\bmu$, and the color-coding uses red for positive deviations, blue for negative ones. 

In order to validate the accuracy of the large-deviations approach, we compare $\bpell$ to the pre-limit conditional expectation of power injections given the failure of line $\ell$, namely
\[\prv_{\eps}=\E[\,\bp_{\eps}\, \vert \,|(\bfl_{\eps})_\ell| \ge 1],\]
which according to Prop.~\ref{prop:mostlikely_injection} converges to $\bpell$ in the limit as $\eps \to 0$.
As a measure of error, we consider, for each line $\ell$,
\begin{align*}
\text{err}(\ell)=\frac{1}{n_s}\sum_{i=1}^{n_s} \Bigl|\frac{(\bp^{(\ell)})_{i}-(\prv_{\eps})_i}{\mu_i}\Bigr|,
\end{align*}
which quantifies the difference between $\bp^{(\ell)}$ and $\prv_{\eps}$,  expressed as a percentages of the nominal values $\bmu$, averaged across all stochastic nodes. We found that, for the same hour as in Fig.~\ref{fig:mostlikely_11am}, the average error across all lines is $\widehat{\text{err}}=\frac{1}{m}\sum_{\ell=1}^m \text{err}(\ell)=0.2\%$, with a maximum value of $2.6\%$, see Fig.~\ref{fig:error_realization_16}. Table~\ref{tab:err_mostlikely} shows that the errors are uniformly small across different hours.

\begin{figure}[!htb]
     \centering
            \includegraphics[width=1\columnwidth]{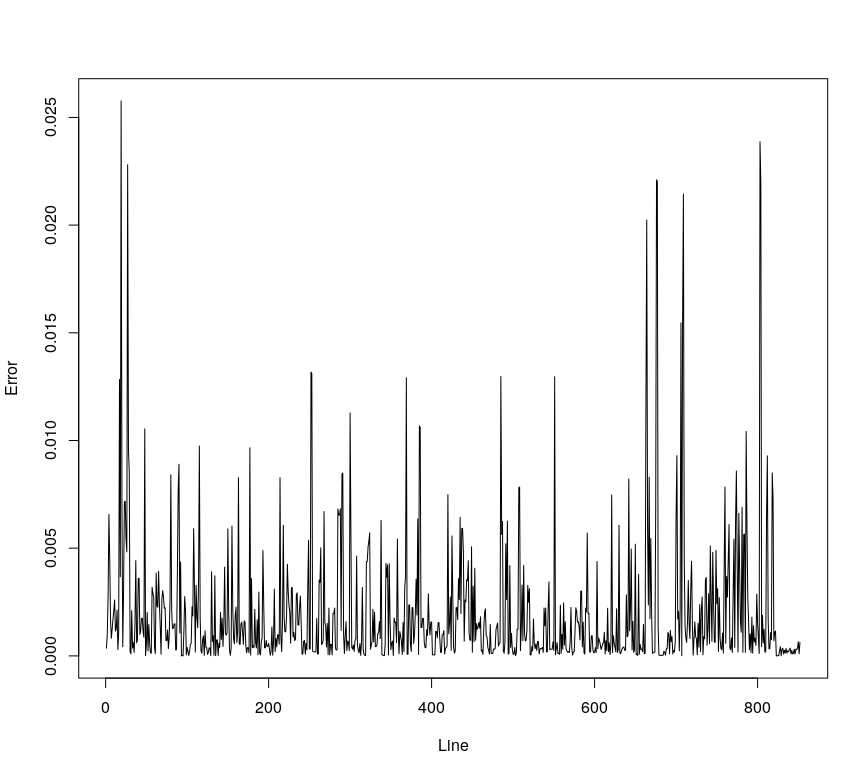}

\caption{\footnotesize  Relative error $\text{err}(\ell)$ at $11$ am, for $\ell=1,\ldots,852$.}\label{fig:error_realization_16}
\end{figure}

\begin{table}[h!]
\begin{tabular}{c|c|c}
\hline
\hline
\text{hour}   \, \,& \,\, $\widehat{\text{err}}$\ \, \, & \, \, $\max \text{err}(\ell)$ \, \\
\hline
4 \text{am}& 0.1\%  & 1.5\% \\
8 \text{am}& 0.4\%  & 4.6\% \\
11 \text{am}& 0.2\%  & 2.6\% \\
4 \text{pm}& 0.1\%  & 2.3\% \\
\hline
\hline
\end{tabular}
\caption{\footnotesize  Average and maximum $\text{err} (\ell)$ for different hours.}
\label{tab:err_mostlikely}
\end{table}

%

%

 \begin{figure}[t!]
    \centering
   \includegraphics[width=1\columnwidth]{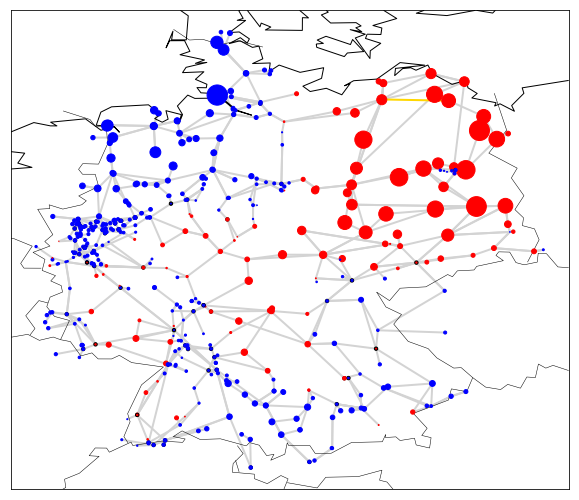}
\caption{\footnotesize Representation of the most likely power injection $\bpell$
causing the isolated failure of the line $\ell = 720$ (orange) at $11$ am. The bus sizes reflect how much $\bpell$ deviates from $\bmu$ (red for positive deviations, blue for negative ones).}
\label{fig:mostlikely_11am}
\end{figure}

\subsection{German network: Failure propagation}
Fig.~\ref{fig:non_local_4pm} shows the emergent isolated failure of line $27$. Such a line is the most likely to fail among those which upon failure do not disconnect the network and trigger subsequent failures; specifically, the failure of line $27$ (in red) causes six more lines $\{k_{1},\ldots,k_{6}\}$ to fail (in orange).
This example shows how the failure spreads non-locally: in particular, lines $316$ and line $602$ in the south of Germany are $394$ Km and $517$ Km far from the original failure of line $27$.
 
\begin{figure}[!htb]
     \centering
  \includegraphics[width=1\columnwidth]{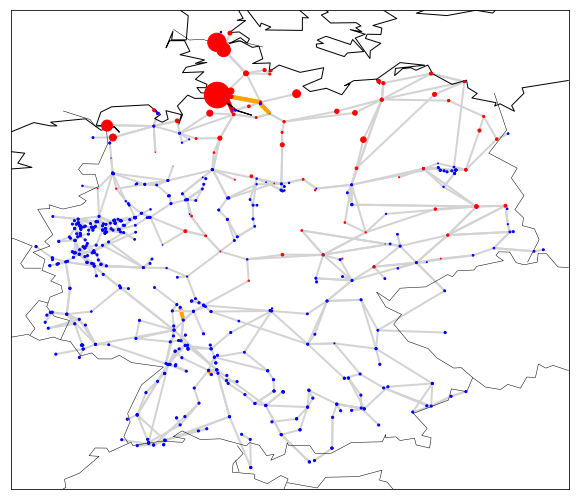}
\caption{\footnotesize
The most likely configuration $p^{(27)}$ that leads to the failure of line $\ell= 27$ (in red), at $4$ pm. The sizes of the buses reflect how much $p^{(27)}$ deviates from $\bmu$ (red for positive deviations, blue for negative ones). The failure of such a line causes, after power redistribution, also the six lines (in orange) to fail.}\label{fig:non_local_4pm}
\end{figure}

In view of Prop.~\ref{prop:second_failure}, the subsequent six failures have been determined by looking at the vector $\bfmostlikely=\widetilde{\bV}^{(\ell)}\bpell$, and checking whether $|\widetilde{f}_k^{(\ell)}|\ge 1$ for each line $k\neq \ell$.  According to Prop.~\ref{prop:second_failure}, the pre-limit conditional probabilities
\[\mathbb{P}(|(\frv_{\eps})_{k_j}|\ge 1\,\rvert\,|(\bfl_{\eps})_\ell|\ge 1)\]
converge exponentially fast to $1$ as $\eps \to 0$, and in particular the cumulative distribution functions
\[\mathbb{P}((\frv_{\eps})_{k_j}\le x\,\rvert\,|(\bfl_{\eps})_\ell|\ge 1)\]
converge to the deterministic distribution $\widetilde{f}_{k_j}^{(\ell)}$.
In order to validate our methodology, we numerically evaluate
\[\mathbb{P}(|(\frv_{\eps})_{k_j}|< 1\, \rvert \, |(\bfl_{\eps})_\ell|\ge 1),\]
for $j=1,\ldots,6$, and found that the probability that all the six lines identified by the large deviations approach actually fail in the pre-limit is equal to
\begin{align*}
&\,\mathbb{P}(|(\frv_{\eps})_{k_j}|\ge 1\,\forall j=1,\ldots,6\, \rvert\, |(\bfl_{\eps})_\ell|\ge 1)\ \notag \\
&\ge \,1-\sum_{j=1}^6\mathbb{P}(|(\frv_{\eps})_{k_j}|< 1\, \rvert\,  |(\bfl_{\eps})_\ell|\ge 1)=0.9987.\label{eq:val_cascade}
\end{align*}
%

\subsection{German network: System security vs System cost}\label{German grid: System security vs System cost}
In order to model a heavily-loaded but not overloaded system, in the OPF we scale the true line limits $C_{\ell}$ by a contingency factor of $\lambda\in (0,1)$. This is a common practice in power engineering that allows room for reactive power flows and stability reserve.

We explore the trade-off between system security and system cost, by varying the contingency factor $\lambda$ in the range $\lambda\in [0.7,1]$. We evaluate system security by means of the large deviations approximation for the failure probability of a given line $\ell$,
\begin{equation}\label{eq:pr}
\Pl(\ell) =\exp (-  I_\ell(\lambda)),
\end{equation}
where we emphasize the dependency on $\lambda$, and we use the average Locational Marginal Price (LMP,~\cite{Schweppe1988}) and the maximum LMP at the grid nodes as metrics of system costs.

Fig.~\ref{fig:LMPvsProb} reports the results corresponding to the same setting as in Fig.~\ref{fig:non_local_4pm}.
From this graph one can, for instance, immediately infer that making line $27$ (the red line in Fig.~\ref{fig:mostlikely_11am}) ten times as safe will roughly cost 1 €/MWh on average, while the increase in cost in terms of maximum price can be much more significant.
This example shows how our large deviations theoretical framework can be a valuable tool to help designing a safe and reliable network at minimal cost. However, as Eq.~\eqref{eq:pr} may not be accurate, more research in this direction is necessary.

%
%
%



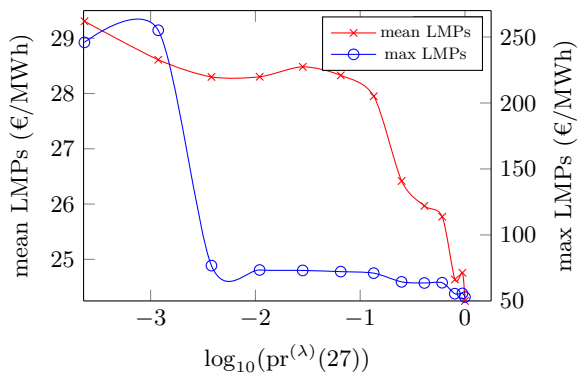
\begin{figure}[!htb]
\centering
\begin{tikzpicture}
\centering
\def\axisdefaultwidth{9cm}
\def\axisdefaultheight{7cm}
\pgfplotsset{
    xmin=-3.64, xmax=0.25
}
\begin{axis}[
  legend style={nodes={scale=0.05, transform shape}},
  axis y line*=left,
  ymin=24.247, ymax=29.5,
  ytick distance=1,
  xlabel=$\log_{10}(\text{pr}^{(\lambda)}(27))$,
  ylabel=mean LMPs (€/MWh),
]
\addplot[smooth,mark=x,red]
  coordinates{
	(-3.631,29.307)
	(-2.927,28.611)
	(-2.419,28.300)
	(-1.959,28.303)
	(-1.548,28.483)
	(-1.185,28.328)
	(-0.870,27.950)
	(-0.604,26.417)
	(-0.387,25.966)
	(-0.217,25.771)
	(-0.096,24.635)
	(-0.024,24.756)
	(0.,24.2472)
}; \label{plot_one}
\addlegendentry{mean LMPs}
\end{axis}
\begin{axis}[
 legend style={nodes={scale=0.7, transform shape}},
  axis y line*=right,
  axis x line=none,
  ymin=50, ymax=270,
  ytick distance=50,
  ylabel=max LMPs (€/MWh)
]
\addlegendimage{/pgfplots/refstyle=plot_one}\addlegendentry{mean LMPs}
\addplot[smooth,mark=o,blue]
  coordinates{
	(-3.63193923,245.91932364)
	(-2.92750962,255.12788312)
	(-2.41942944,76.77073475)
	(-1.95973785,73.43800017)
	(-1.54843484,73.10351827)
	(-1.18552043,72.19827009)
	(-0.8709946,71.05671876)
	(-0.60485736,64.45094029)
	(-0.38710871,63.59735028)
	(-0.21774865,63.78439619)
	(-0.09677718,55.50065369)
	(-0.02419429,55.70097569)
	(0.,52.93556707)
}; \label{plot_two}
\addlegendentry{max LMPs}
\end{axis}
\end{tikzpicture}
\caption{\footnotesize Average LMP (scale on left) and Maximum LMP (right) vs.
$\log_{10} (\Pl(27)) =\log_{10} (\exp (- I_{27}))$, for the German network at $4$pm. }
\label{fig:LMPvsProb}
\end{figure}

Reducing the security margin does not only influence the average LMPs and system costs, but also their geographical distribution. Fig.~\ref{fig:LMPgeo} shows geographically accurate LMPs for two values of $\lambda$, one corresponding to a low effective limit/large security margin system $(\lambda=0.7)$ and the other to a large effective limit/low security margin system $(\lambda=0.95)$. We can see how to a more conservative system corresponds LMPs which are larger especially in the south and south-west part of Germany, while in northern Germany the difference is less pronounced. Quoting \cite{PyPSA2017}, this phenomenon can be explained by the fact that ``transmission bottlenecks in the middle of Germany prevent the transportation of this cheap electricity to the South, where more expensive conventional generators set the price".

\begin{figure}[!h]
\centering
\includegraphics[width=\columnwidth]{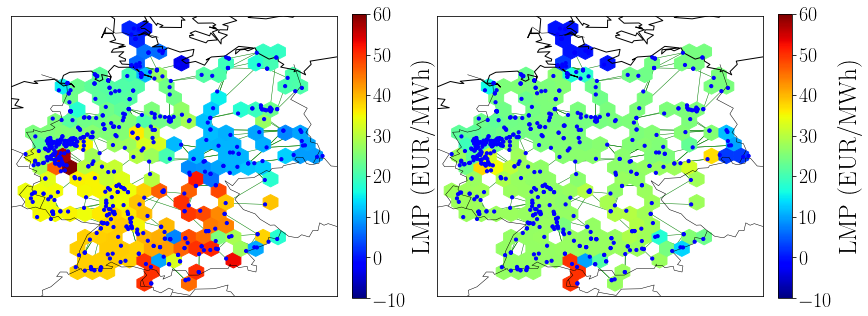}
\caption{\footnotesize Geographical distribution of LMPs for $\lambda=0.7$ (left) and $\lambda=0.95$ (right), at $4$pm.}
\label{fig:LMPgeo}
\end{figure}

Furthermore, Fig.~\ref{fig:number_lines} shows that reducing the system security margin does not only increases the likelihood of an overload, but it also increases the number of lines with a large enough overload probability.

\begin{figure}[!h]
\centering
\includegraphics[width=0.95\columnwidth]{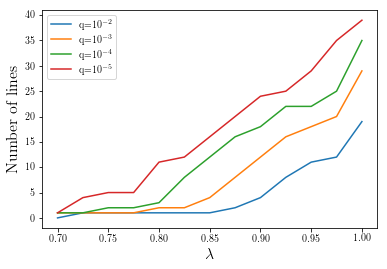}
\caption{\footnotesize     Number of lines $\ell$ with overload probability $\text{pr}^{(\lambda)}(\ell)\ge q$.}
\label{fig:number_lines}
\end{figure}

%

\FloatBarrier

\section{Cascading analysis: classical versus emergent failures}
As illustrated earlier, the most likely power injections configuration leading to the emergent failure of a given line can be used in combination with the power flow redistribution rules to generate the failures triggered by that initial scenario. By repeating this procedure for all lines, one can obtain insightful statistics of the first two stages of emergent cascading failures (ec) and compare them with those of classical cascading failures (cc), obtained using nominal power injection values rather than the most likely ones and deterministic removal of the initial failing line. We perform numerical experiments using IEEE test grids. Since several IEEE test-cases do not report realistic transmission limits, line thresholds are taken to be proportional to the average absolute power flow on the corresponding lines, i.e., $C_\ell = (1+\alpha) |\nu_\ell|$, where $\nu_{\ell}$ is a nominal value provided in the dataset, $\alpha=0.25$ and $\Sp$ is the identity matrix.

\begin{table}[!h]
\begin{tabular}{c|c|c|c|c}
\hline
\hline
Graph & \% joint failures & $\, \, \E(F_1^{\textrm{ec}}) \, \,$ & $\, \, \E(F_2^{\textrm{ec}}) \, \,$ & $\, \,\E(F_2^{\textrm{cc}}) \, \,$\\
\hline
IEEE14 & 65.0\% & 4.40 & 8.40 & 4.95 \\
IEEE30 & 97.6\% & 3.73 & 9.88 & 4.95\\
IEEE39 & 80.4\% & 4.78 & 11.39 & 4.85\\
IEEE57 & 88.5\% & 8.00 & 19.00 & 10.44\\
IEEE96 & 72.2\% & 6.70 & 21.47 & 7.31\\
IEEE118 & 91.6\% & 10.40 & 24.53 & 7.56\\
IEEE300 & 87.0\% & 18.13 & 39.19 & 7.42\\
\hline
\hline
\end{tabular}
\caption{\footnotesize     Percentage of joint failures in emergent cascades and average number of failed lines $F_1$ up to stage 1 and $F_2$ up to stage 2 for emergent cascades (ec) and classical cascades (cc) for some IEEE test systems.}
\label{tab:one}
\end{table}

As shown in Table~\ref{tab:one}, emergent cascades have a very high percentage of joint failures and an average number of failures in the first cascade stage much larger than one (in classical cascades only one line is removed in the first cascade stage). Furthermore, the expected total number of failed lines up to the second cascade stage is significantly larger for emergent cascades than for classical cascades. Lastly, failures propagate in emergent cascades on average a bit less far than in classical cascades, as illustrated by the statistics of the failure jumping distance in Table~\ref{tab:two}.

\begin{table}[!h]
\begin{tabular}{c|c|c|c|c}
\hline
\hline
Graph & $\, \,\E(D^{\textrm{ec}}) \,\,$ & $\,\, \E(D^{\textrm{cc}}) \,\,$ & $\, \,c_v(D^{\textrm{ec}}) \,\,$ & $\, \,c_v(D^{\textrm{cc}}) \,\,$\\
\hline
IEEE14 & 0.388 & 0.987 & 0.600 & 1.050\\
IEEE30 & 0.754 & 1.198 & 0.879 & 1.115\\
IEEE39 & 0.898 & 1.633 & 0.891 & 1.149\\
IEEE57 & 1.210 & 2.507 & 0.863 & 1.415\\
IEEE96 & 1.450 & 1.781 & 0.879 & 0.946\\
IEEE118 & 0.679 & 1.638 & 0.745 & 1.169\\
IEEE300 & 1.408 & 2.580 & 0.806 & 1.081\\
\hline
\hline
\end{tabular}
\caption{\footnotesize     Average and coefficient of variation of the failure jumping distance $D$ in stage 2 both for emergent cascades (ec) and classical cascades (cc). The distance between two lines is measured as the shortest path between any of their endpoints.}
\label{tab:two}
\end{table}

Our approach also gives a constructive way to build the so-called ``influence graph'' \cite{Hines2013,Hines2015,Qi2015}, in which a directed edge connects lines $\ell$ and $\ell'$ if the failure of the line $\ell$ triggers (simultaneously or after redistribution) that of line $\ell'$. Fig.~\ref{fig:influencegraph118} shows an example of influence graph built using our large deviations approach. The \textit{cliques} of the influence graph (i.e., its maximal fully connected subgraphs) can then be used to identify clusters of cosusceptable lines~\cite{YangNishikawaMotter2017}, which are the lines that statistically fail often in the same cascade event.

\begin{figure}[!h]
	\includegraphics[scale=0.35]{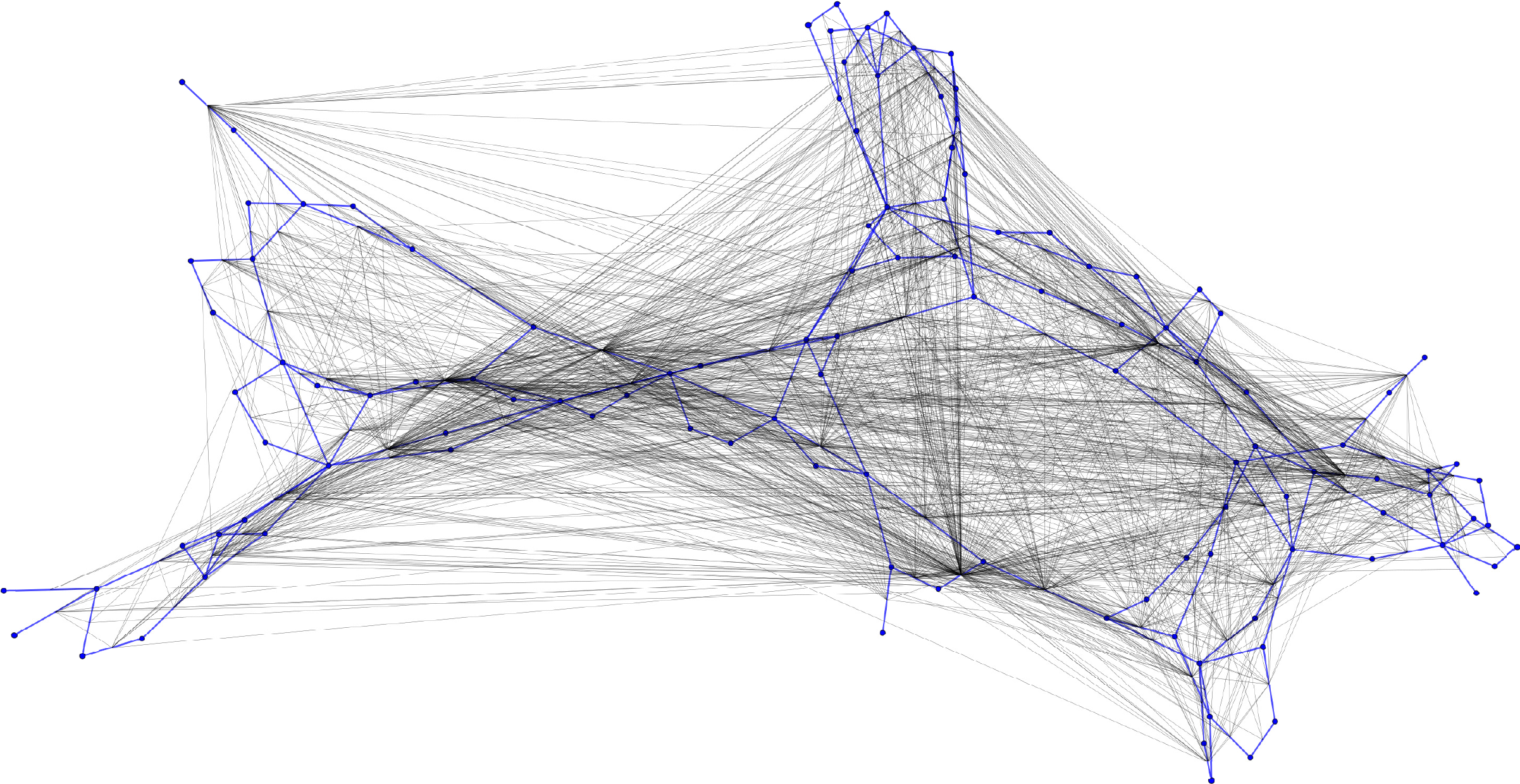}
	\caption{\footnotesize     The influence graph of the IEEE 118-bus test system (in black) built using the first two stages of all cascade realizations has a deeply different structure than the original network (in blue).}
	\label{fig:influencegraph118}
\end{figure}

\begin{thebibliography}{70}
\bibitem [{NZZ()}]{NZZ17sm}%
  \BibitemOpen
  \href@noop {} {\enquote {\bibinfo {title} {See {S}upplemental {M}aterial at
  ... for all the details about the power flow model, the large deviations
  principle for failures and the numerics for the {G}erman network, which
  includes refs.~[53-70]},}\ }\BibitemShut {NoStop}%
\end{thebibliography}
\end{document}